%% file: manuscript.tex
\documentclass[letterpaper, 12pt, journal, onecolumn]{ieeeconf}      

\IEEEoverridecommandlockouts                              

\usepackage{graphicx}          
\usepackage{booktabs}
\usepackage{subcaption}
\usepackage{amsmath} 
\usepackage{amssymb}  
\usepackage{mathtools}
\usepackage{color}
\usepackage{caption}
\usepackage{multirow}
\usepackage{url}
\usepackage[compress]{cite}
\usepackage[short]{optidef}
\newtheorem{thm}{Theorem}
\newtheorem{prop}[thm]{Proposition}
\newtheorem{lem}[thm]{Lemma}

\newtheorem{defn}[thm]{Definition}
\newtheorem{assum}[thm]{Assumption}
\newtheorem{rem}[thm]{Remark}
\newenvironment{pf*}[1][Proof]{\vspace{1.5ex}
	\it #1  \rm}
{\hfill \footnotesize{$\blacksquare$}\vspace{2ex}}

\newcommand{\blu}[1]{\textcolor{black}{#1}}
\newcommand{\blue}[1]{\textcolor{black}{#1}}
\title{Strong nonlinear detectability and  moving horizon estimation for nonlinear  systems with unknown inputs}

\author{Yang Guo $^{a}$, Jaime A. Moreno $^{b}$ and Stefan Streif $^{a}$
\thanks{$^{a}$The authors are with  Professorship Automatic Control and System Dynamics, Chemnitz University of Technology, 09107 Chemnitz, Germany
        {\tt\small \{yang.guo, stefan.streif\}@etit.tu-chemnitz.de}}%
\thanks{$^{b}$The author is with  Institute of Engineering, Universidad Nacional Aut\'onoma de M\'exico, 04510 Mexico City, Mexico 
        {\tt\small JMorenoP@iingen.unam.mx}}%
}

\begin{document}
\maketitle
\thispagestyle{empty}
\pagestyle{empty}



\begin{abstract} 

This paper investigates state estimation for general nonlinear discrete-time systems subject to measurement noise and possibly unbounded unknown inputs. To approach this problem, we first propose the concept of strong nonlinear detectability, which proves to be a  nonlinear counterpart to  strong$^{\ast}$ detectability. Based on the proposed detectability notion, a \textcolor{black}{robust} unknown input state estimator (UISE), which yields bounded estimation error even for unbounded unknown inputs,  \blue{is designed via a moving horizon estimation strategy}.  Next,   
we tighten this detectability notion to develop a two-stage MHE-based UISE that is  computationally more efficient. 
	In a  simulation example, both variants of MHE-based UISEs are compared with a conventional MHE  to illustrate the merits of the proposed methods.    

\end{abstract}

\section{Introduction}

State estimation is a prerequisite for numerous control applications and is particularly challenging for nonlinear systems with  uncertainties and unknown inputs. In addition, complex components of the model  may be treated as unknown inputs in order to mitigate the difficulties associated with the design and implementation of estimators.
  Common strategies to ensure bounded estimation error under bounded unknown inputs use robust moving horizon estimators (MHE) \cite{alessandri2008moving, muller2017nonlinear, ji2015, knufer2023nonlinear} or robust nonlinear observers \cite{einicke1999robust, besanccon2003high, alessandri2004observer, tran2023arbitrarily}.  Nevertheless,  estimation accuracy can deteriorate considerably  when unknown inputs become large or even unbounded, as estimation errors are influenced by  these unknown inputs. 
  
To achieve better  estimation performance,  the design of  so-called  unknown input observers (UIOs) has been investigated for decades. A UIO is a dynamical system that effectively decouples the state estimation from the unknown inputs,  enabling exact estimation \textcolor{black}{from measured inputs and outputs}  in noise-free cases.       
It has also found   a variety of applications, such as sensor fault detection \cite{zarei2014robust}, and cyber attack detection in networked systems \cite{gallo2020distributed}.   For linear systems, the design of UIOs has been intensively studied \cite{hautus1983,tranninger2022, valcher1999, darouach2009complements, disaro2024, ZHU2023, zou2020moving, moreno2023high}. For example, Hautus proposes the notion of strong$^\ast$ detectability\footnote{\blu{Hautus uses the superscript $\ast$ to differentiate this   from the notion of  strong detectability \cite{hautus1983}. For discrete-time LTIs, both notions are equivalent; for continuous-time LTIs, strong$^\ast$ detectability implies  strong detectability.}   } in \cite{hautus1983} as the   necessary and sufficient  condition for the existence of UIOs for  continuous- and discrete-time linear time-invariant systems (LTIs). This result is generalized to continuous-time \blu{linear} time-varying systems in \cite{tranninger2022}.   
The work \cite{disaro2024} proposes a  data-driven design approach for  UIOs.  Simultaneous state estimation and unknown input reconstruction is considered in \cite{ZHU2023} via interval observers.  In \cite{zou2020moving}, state estimation under dynamic quantization effects is addressed using MHE strategy.

For a certain class of nonlinear continuous-time systems with bounded unknown inputs, homogeneity (or sliding mode) theory has been widely employed to design  UIOs, cf. \cite{floquet2007super, zhou2014high,   Caamal2023}.   In the case of  unbounded unknown inputs,    \blu{UIOs  can be constructed using dissipativity \cite{rocha2005} for a class of continuous-time  Lipschitz  systems or  a  projection approach \cite{alenezi2020state} for  discrete-time Lur'e-like systems, or   geometric methods \cite{hammouri2010unknown} for continuous-time state affine systems.  In  \cite{rocha2005} and \cite{alenezi2020state},         restrictive structural assumptions on systems are stated as sufficient conditions  for the existence of UIOs. In contrast,  \cite{hammouri2010unknown} provides sufficient and necessary conditions, which, together with the design methodology, are nevertheless mainly conceptual and not directly suitable for practical application.   By state augmentation,    the construction of UIOs for a class of Lipschitz systems can  be recast as   observer design for singular systems under sufficient conditions, cf. \cite{delshad2016robust, darouach2008observers}. }  Notably, while there exists a large body of literature on continuous-time UIO design,  the discrete-time counterpart, despite its practical relevance,  has received much less attention, in particular for nonlinear systems. 
\blu{Furthermore, there are no results regarding the design  of state estimators   for general nonlinear systems with  unknown inputs;  a non-conservative condition for the existence of such estimators or,  more specifically,  a  generalization of strong$^{\ast}$ detectability to nonlinear systems  is not  available.  }

 In the present paper,  we  consider  state estimation for \emph{general} nonlinear discrete-time systems subject to  bounded measurement noise and possibly unbounded unknown inputs.  
 The  main contributions are threefold:
 \begin{itemize}
 	\item  We  propose a concept of strong nonlinear detectability and provide its equivalent Lyapunov-like characterization.  \blu{This concept proves to be 
 	a nonlinear counterpart to strong$^\ast$ detectability \cite{hautus1983}. }
 	\item We present a design framework for \textcolor{black}{robust}  unknown input state estimators (UISEs) built on a  moving horizon estimation scheme and \blu{the Lypapunov characterization of strong nonlinear detectability.} The  MHE-based \textcolor{black}{robust}  UISE features convergent estimation errors for vanishing measurement noise \blue{and output error} even when  unknown inputs are unbounded.    This extends recent developments  in robust MHE from \cite{allan2021robust,   knufer2023nonlinear, schiller2023} to nonlinear systems with \blue{unbounded disturbances}. 
 	\item   By leveraging a stronger condition \textcolor{black}{of strong nonlinear detectability and a  residual model,  we propose  a  two-stage MHE-based UISE using past measured outputs and current estimated outputs,}  reducing the online computational complexity.
 \end{itemize}
 
 The paper is organized as follows.  The problem setting is given in Section~\ref{sec:setting}. We propose  the concept of strong nonlinear detectability and  MHE-based \textcolor{black}{robust} UISE design in Section~\ref{sec:UISE_full} and  a computationally  efficient MHE-based UISE  in Section~\ref{sec:UISE_red}. In Section~\ref{sec:sim}, the developed MHE-based UISEs  are  implemented on a  simulation example and compared with a conventional  approach.  Conclusions and directions for future research are outlined in Section~\ref{sec:conclusion}.

\textit{Notation:} We use $\mathbb{I}_{[a,b]}$ to denote the set of integers in the interval $[a,b]$. We denote the infinite sequence $(x_i)^\infty_{i=0}$ with $x_i \in \mathbb{R}^n$, $n\in \mathbb{N}$,   by $\boldsymbol{x}$ and use $\boldsymbol{x}_{[a,b]}$ to represent the finite sequence $(x_i)^{b}_{i=a}$. The notion $\|\boldsymbol{x}\|_{[a,b]}$ is defined to be $\max_{i\in \mathbb{I}_{[a,b]}} \|x_i\|$ with $\|\cdot\|$ representing the Euclidean norm. We use   $\lfloor \cdot \rfloor$ as the floor function defined on $\mathbb{R}$.   Let  $\text{col}(x_1, \ldots, x_n )$ denote the vector $(x^\top_1, \ldots, x^\top_n)^\top$.  \textcolor{black}{We abbreviate the Cartesian product of a family of sets $\{ \mathbb{X} \}_{i \in \mathbb{I}_{[0, \infty)}}$  by $ \mathbb{X}^\infty$. }

A function $\alpha: [0, \infty) \to [0, \infty)$ is in  class  $\mathcal{K}$  if it is continuous, strictly increasing, and $\alpha(0)=0$. It is in class $\mathcal{K}_\infty$ if, in addition,  it is unbounded. It is in class \textit{generalized} $\mathcal{K}$, termed as $\mathcal{GK}$,  if it is continuous, $\alpha(0)=0$, and satisfies $\alpha(r_1)>\alpha(r_2)$ for $r_1>r_2$ with $\alpha(r_1)>0$ as well as  $\alpha(r_1)=\alpha(r_2)$ for $r_1>r_2$ with  $\alpha(r_1)=0$.  A function $\beta: [0, \infty)\times \mathbb{I}_{[0, \infty)} \to [0, \infty)$ is in class $\mathcal{KL}$ if $\beta(\cdot, k) \in \mathcal{K}$ for a fixed $k \in \mathbb{I}_{[0, \infty)}$,  and $\beta(r, \cdot)$ is non-increasing and satisfies $\lim_{k\to \infty} \beta(r,k)=0$ for a fixed $r \in [0, \infty)$. It is in class \textit{generalized}  $\mathcal{KL}$, termed as $\mathcal{GKL}$,  if $\beta(\cdot, k) \in \mathcal{GK}$ for a fixed $k \in \mathbb{I}_{[0, \infty)}$ and  $\beta(r, \cdot)$ is non-increasing and satisfies $\lim_{k\to T} \beta(r,k)=0$ for some finite $T \in \mathbb{I}_{[0, \infty)}$ and a fixed $r \in [0, \infty)$.   A function $\beta \in \mathcal{KL}$ is summable if there exists an $\alpha \in \mathcal{K}$ such that $\sum^{\infty}_{k=0} \beta(r, k) \leq \alpha(r)$ holds for all $r \in [0, \infty)$.




\section{Problem setup} \label{sec:setting}
We consider the following discrete-time nonlinear system
\begin{subequations}
	\label{eq:sys}
	\begin{align}
		x_{k+1} &= f(x_k,u_k,w_k), \label{eq:sys_1} \\
	y_k &= h(x_k, v_k), \label{eq:sys_2}
	\end{align}
\end{subequations}
with the state $x_k\in \mathbb{X} \subseteq \mathbb{R}^n$, the  control or known  input $u_k\in \mathbb{U} \subseteq \mathbb{R}^m$,  the  possibly unbounded unknown input $w_k\in  \mathbb{R}^q$,  the  bounded measurement noise\footnote{In  the  MHE literature, cf. \cite{cai2008input,alessandri2008moving},    the term   ``noise" or ``disturbance" is typically used to refer to   a bounded unknown input.} $v_k \in \mathbb{V} \subset \mathbb{R}^r$,    as well as the measured output  $y_k\in \mathbb{Y} \subseteq \mathbb{R}^p$.
Throughout this work, we assume that $f$ and $h$ are  continuously differentiable,   and  the domain of trajectories $\mathbb{Z}:= \mathbb{X} \times \mathbb{U} \times   \mathbb{R}^q \times \mathbb{V} \times \mathbb{Y}$ is known  a priori and closed.

Analogous to \cite[Def.~2.2]{allan2021}, we employ the following input-output map to formalize state estimators comprising  observers  as well as estimation algorithms including the full information estimator (FIE) and MHE.

\begin{defn}[State estimator]\label{def:estimator}
	A state estimator for the system~\eqref{eq:sys} is a sequence of functions, defined as
	\begin{equation} \label{eq:Psi_k}
		\hat{x}_k= \Psi_k \left(\bar{x}_0, \bar{\boldsymbol{v}}_{[0,k]},  \boldsymbol{u}_{[0,k-1]},  \textcolor{black}{\bar{\boldsymbol{y}}_{[0,k]}} 
		\right)
	\end{equation}
	for all $k \in \mathbb{I}_{[0, \infty)}$, where $\hat{x}_k$ is the state estimate, $\bar{x}_0 \in \mathbb{X}$ is \textcolor{black}{the prior of the initial state,   $\bar{v}_j \in \mathbb{V}$ and $\bar{y}_j \in \mathbb{Y}$ are estimates of the measurement noise and the measured output} respectively,   and   $ u_j \in \mathbb{U} $ is the control input  of the system~\eqref{eq:sys} at time $j$.  
\end{defn}

\textcolor{black}{The estimated output  $\bar{y}_j$ supplied to the estimator may differ from the measured output $y_j$. In applications, however,   $\bar{y}_j$ is typically chosen to be  $y_j$.  The more general definition presented here will be crucial to provide a nonlinear generalization of  strong$^\ast$  detectability.}

The main goal of this article is to systematically design a state estimator that possesses the property specified by \textcolor{black}{one of} the following  \textcolor{black}{definitions  depending on whether output measurements  are supplied to estimators.  }  

\begin{defn}[\textcolor{black}{Robust} unknown input state estimator] \label{def:RGAS}
	A state estimator  is called a \textcolor{black}{robust} unknown input state  estimator (UISE),  if there exist $\mathcal{K}_{\infty}$-function $\beta$,    $\mathcal{KL}$-function $ \beta_0$ and summable  $\mathcal{KL}$-functions  $\beta_v$ \textcolor{black}{and $\beta_y$} such that the computed estimate $\hat{x}_k$  satisfies
	\begin{equation}\label{eq:RGES}
		\beta(\|x_k-\hat{x}_k\|) \leq  \beta_0( \|x_0-\bar{x}_0 \|, k) +  
		 \textcolor{black}{\sum_{i=0}^{k}\bigl( \beta_v (\|v_{k-i}-\bar{v}_{k-i}\| , i ) +  \beta_y (\|y_{k-i}-\bar{y}_{k-i}\| , i ) \bigr)}  	
\end{equation}  
	for all  $k\in \mathbb{I}_{[0, \infty)}$, any trajectory $(\boldsymbol{x}, \boldsymbol{u}, \boldsymbol{w}, \boldsymbol{v}, \boldsymbol{y}) $ $ \in \mathbb{Z}^\infty$ of \eqref{eq:sys},  \textcolor{black}{any $\bar{x}_0 \in \mathbb{X}$, $\bar{v}_j \in \mathbb{V}$ and $\bar{y}_j \in \mathbb{Y}$, $j\in \mathbb{I}_{[0,k]}$}.
\end{defn}
\begin{defn}[Unknown input state estimator] \label{def:RGAS_2}
	A state estimator  is called an unknown input state  estimator (UISE),  if the computed estimate $\hat{x}_k$  satisfies \eqref{eq:RGES}   
	for all  $k\in \mathbb{I}_{[0, \infty)}$, any trajectory $(\boldsymbol{x}, \boldsymbol{u}, \boldsymbol{w}, \boldsymbol{v}, \boldsymbol{y}) $ $ \in \mathbb{Z}^\infty$ of \eqref{eq:sys},  \textcolor{black}{any $\bar{x}_0 \in \mathbb{X}$, $\bar{v}_j \in \mathbb{V}$, $j\in \mathbb{I}_{[0,k]}$ and $\bar{\boldsymbol{y}}_{[0, k]}=\boldsymbol{y}_{[0, k]}$}.
\end{defn}

 \textcolor{black}{Every robust UISE is a UISE, and the converse does not hold in general\footnote{\textcolor{black}{For UIOs in linear cases, the estimation error dynamics are linear, cf. \cite{darouach2009complements},  and hence asymptotic  stability   is equivalent to input-to-state stability (ISS).   Thus, Definition~\ref{def:RGAS} and~\ref{def:RGAS_2} become equivalent.}}, since  robustness against mismatch in outputs is not required in  Definition~\ref{def:RGAS_2}.}
Departing from the  notion of incremental ISS  used as the definition of  estimator \textcolor{black}{(robust)} stability  in the MHE literature \cite{ji2015, rawlings2020model,  knufer2020,allan2021}, the state estimation error $x_k-\hat{x}_k$ for a UISE converges to the origin as $k\to \infty$ with vanishing  $v_k-\bar{v}_k$ \textcolor{black}{and $y_k-\bar{y}_k$}, regardless of  unknown inputs $w_k$.  

\section{Unknown input state estimators with full-order models}\label{sec:UISE_full}

\subsection{Strong nonlinear detectability}

In the following definition,  we present a  notion of strong nonlinear detectability,
which serves as a basis for  the design of  UISEs.   

\begin{defn} [Strong nonlinear detectability] \label{def:strong_detect}
The system \eqref{eq:sys} is strongly nonlinearly detectable if there exist  $\alpha \in \mathcal{K}_{\infty}$,  $\alpha_0\in \mathcal{KL}$ and summable $\mathcal{KL}$-functions $ \alpha_v, \alpha_y $  such that  
\begin{equation}
\label{eq:strong_detect}
\begin{aligned}
& \alpha(\|x_k-\widetilde{x}_k\|) \leq  \alpha_0( \|x_0-\widetilde{x}_0 \|, k ) + \sum_{i=0}^{k} \bigl( \alpha_v (\|v_{k-i}-\widetilde{v}_{k-i}\| , i ) + \alpha_y (\|y_{k-i}-\widetilde{y}_{k-i}\| , i ) \bigr), 
\end{aligned}
\end{equation}
for all $k \in \mathbb{I}_{[0, \infty)}$ and any solutions $(\boldsymbol{x},  \boldsymbol{u}, \boldsymbol{w},  \boldsymbol{v}, \boldsymbol{y}), $
$(\widetilde{\boldsymbol{x}},  \boldsymbol{u}, \widetilde{\boldsymbol{w}},  \widetilde{\boldsymbol{v}}, \widetilde{\boldsymbol{y}}) \in \mathbb{Z}^\infty$ of  \eqref{eq:sys}. 
\end{defn}

One of the key differences from  (uniform) incremental input-output-to-state stability (i-IOSS)  \cite[Definition~2.4]{allan2021}, \cite[Definition~4]{knufer2023nonlinear}, which has recently been termed the standard notion of nonlinear  detectability,  is that   the discrepancy between  two unknown input trajectories $\boldsymbol{w}-\widetilde{\boldsymbol{w}}$ is not encompassed within the  notion~\eqref{eq:strong_detect}. Another difference is that the notion~\eqref{eq:strong_detect} incorporates not only past noise and  output measurements but also current ones,  analogous to  the non-incremental IOSS in  \cite[Defintion~3.4]{cai2008input}.  Hence,  it is more general than  notions formulated solely in terms of  past noise and  measurements. This  permits the use of  current measurements for estimation. 


In what follows, we provide an \emph{equivalent} Lyapunov-like characterization of strong nonlinear detectability, facilitating the numerical verification of the proposed concept of detectability. 
\begin{thm}\label{theo:Lyapunov}
	The system~\eqref{eq:sys} is strongly  nonlinearly detectable according to Definition~\ref{def:strong_detect}     if and only if there exist  $\alpha_1, \alpha_2 \in \mathcal{K}_{\infty}$, $\sigma_v,   \sigma_y \in \mathcal{K}$, $\mu \in (0,1)$ and an incremental storage function $V(x, \widetilde{x}): \mathbb{X} \times \mathbb{X} \to \mathbb{R}_{\geq 0}$  such that 
	\begin{subequations} \label{eq:Lyapunov}
		\noindent
		\begin{align}
			&\alpha_1(\|x-\widetilde{x}\|) \leq V(x, \widetilde{x}) \leq \alpha_2(\|x-\widetilde{x}\|) \label{eq:Lyapunov_1} \\
			\begin{split}\label{eq:Lyapunov_2} 
				& V(f(x,u,w), f(\widetilde{x}, u, \widetilde{w}))  \\
				& \leq \mu V(x, \widetilde{x}) + \sigma_{v}(\|v- \widetilde{v}\|)   + \sigma_{v}(\|v^+- \widetilde{v}^+\|) 
                + \sigma_y( \|y-\widetilde{y}\|) + \sigma_{y}(\|y^+-\widetilde{y}^+\|)  
			\end{split}
		\end{align}
	\end{subequations}
	for all $(x,u,w, v, y)$, $(\widetilde{x},u,\widetilde{w}, \widetilde{v}, \widetilde{y})\in \mathbb{Z}$ and $v^+,  \widetilde{v}^+ \in  \mathbb{V} $ with  $y=h(x,v)$, $ \widetilde{y}=h(\widetilde{x},  \widetilde{v})$, $y^+=h(f(x,u,w), v^+) \in \mathbb{Y}$ and $\widetilde{y}^+= h(f(\widetilde{x}, u, \widetilde{w}), \widetilde{v}^+) \in \mathbb{Y}$. 
\end{thm}

In the following, we show that strong nonlinear detectability  is  sufficient and necessary for the existence of  a \textcolor{black}{robust} UISE specified in Definition~\ref{def:RGAS}. To this end, we first construct a FIE to show the sufficiency. Specifically, 
given the priors $\bar{x}_0 \in \mathbb{X}$,    $\bar{\boldsymbol{v}}_{[0,k]} \in \mathbb{V}^{k+1}$, \textcolor{black}{$\boldsymbol{\bar{y}}_{[0,k]} \in \mathbb{Y}^{k+1}$} and  all past input signals $\boldsymbol{u}_{[0,k-1]} \in \mathbb{U}^k$,      the estimate  $\hat{x}^\star_{k|k}$ for the state $x_k$  at each time $k \in \mathbb{I}_{[0, \infty)}$ is obtained by solving the following optimization problem
\begin{mini!}{\substack{\hat{x}_{0 |k }, \hat{\boldsymbol{v}}_{\cdot |k } \\ \hat{\boldsymbol{w}}_{\cdot |k } }}{ J_k(\hat{x}_{0 |k }, \hat{\boldsymbol{v}}_{\cdot |k }, \hat{\boldsymbol{y}}_{\cdot |k } )} {\label{eq:op_FIE}}{}
	\addConstraint{ \hat{x}_{j+1|k}}{=  f(\hat{x}_{j|k}, u_j, \hat{w}_{j|k}), \ j\in \mathbb{I}_{[0, k-1]}\label{eq:op_FIE_input}}
	\addConstraint{\hat{y}_{j|k}}{= h(\hat{x}_{j|k}, \hat{v}_{j|k})\label{eq:op_FIE_output}}
	\addConstraint{ \hat{x}_{j|k}}{\in \mathbb{X}, \  \hat{v}_{j|k} \in \mathbb{V}}
	\addConstraint{\hat{w}_{j|k}}{\in \mathbb{R}^q,   \ \hat{y}_{j|k}\in \mathbb{Y}, \ j\in \mathbb{I}_{[0,k]}\label{eq:op_FIE_output_2}}
\end{mini!}
with  noise estimates $\hat{\boldsymbol{v}}_{\cdot|k}=(\hat{v}_{j|k})^k_{j=0}$, unknown input estimates  $\hat{\boldsymbol{w}}_{\cdot|k}=(\hat{w}_{j|k})^k_{j=0} $    and output estimates $\hat{\boldsymbol{y}}_{\cdot|k}=(\hat{y}_{j|k})^k_{j=0}$ and the initial  state estimate $\hat{x}_{0|k}$.

The cost function therein is defined by
\begin{equation}
	\label{eq:cost_FIE_full}
	\begin{aligned}
& J_{k}(\hat{x}_{0|k },   \hat{\boldsymbol{v}}_{\cdot |k }, \hat{\boldsymbol{y}}_{\cdot |k }) = \alpha_0(2\| \hat{x}_{0|k}-\bar{x}_{0} \|, k) +  \sum_{j=0}^{k}\bigl( \alpha_v( 2\| \hat{v}_{k-j|k}-\bar{v}_{k-j}\|,j) + \alpha_y(\textcolor{black}{2\|\bar{y}_{k-j}-\hat{y}_{k-j|k}\|, j)}     \bigr)
	\end{aligned}
\end{equation}  
with $\alpha_0, \alpha_v, \alpha_y$ from Definition~\ref{def:strong_detect}. 
\begin{lem}\label{lem:FIE}
	The FIE formulated by \eqref{eq:op_FIE} is a \textcolor{black}{robust} UISE if the system~\eqref{eq:sys} is strongly nonlinearly detectable. 
\end{lem}

\begin{proof}	
		Let   $\hat{x}_k=\hat{x}^\star_{k|k}$. By applying \eqref{eq:strong_detect}, the optimality of $J_k( \hat{x}^\star_{0 |k },  \hat{\boldsymbol{v}}^\star_{\cdot |k }, \hat{\boldsymbol{y}}^\star_{\cdot |k } )$ and the  weak triangle inequality \cite[(6)]{jiang1994}, we get
	\begin{align*}
	 &	\alpha(\|x_k - \hat{x}_{k} \|)  \leq \alpha_0(2\|x_0-\bar{x}_0\|, k) +  \sum_{j=0}^{k} \bigl(\alpha_v (2\| v_{k-j}- \bar{v}_{k-j} \|, j) + \textcolor{black}{    \alpha_y (2\| y_{k-j}- \bar{y}_{k-j} \|, j)}\bigr)  \\
	 & +J_k( \hat{x}^\star_{0 |k },  \hat{\boldsymbol{v}}^\star_{\cdot |k }, \hat{\boldsymbol{y}}^\star_{\cdot |k } )  \\
     & \leq \alpha_0(2\|x_0-\bar{x}_0\|, k)  +\sum_{j=0}^{k} \big(\alpha_v (2\| v_{k-j}-\bar{v}_{k-j} \|, j) + \textcolor{black}{\alpha_y (2\| y_{k-j}-\bar{y}_{k-j} \|, j) }  \big) + J_k( x_0,  \boldsymbol{v}_{[0,k] }, \boldsymbol{y}_{[0, k] } ) \\
	& =  2\alpha_0(2\|x_0-\bar{x}_{0}\|, k) + \sum_{j=0}^{k} \left(2  \alpha_v (2\| v_{k-j}-\bar{v}_{k-j} \|, j) + \textcolor{black}{2  \alpha_y (2\| y_{k-j}-\bar{y}_{k-j} \|, j)}  \right)
	\end{align*}
	for $k \in \mathbb{I}_{[0, \infty)}$ and any   $(\boldsymbol{x},  \boldsymbol{u}, \boldsymbol{w},  \boldsymbol{v}, \boldsymbol{y})$ of \eqref{eq:sys}. 
		Replacing $\alpha(r)$, $2\alpha_0(2r, k)$, $2 \alpha_v(2r,k)$ \textcolor{black}{and  $2 \alpha_y(2r,k)$} \blue{respectively} by $\beta(r)$,  $\beta_0(r,k)$, $\beta_v(r, k)$ \textcolor{black}{and $\beta_y(r, k)$ } in the above inequality completes the proof.  
\end{proof}

\blue{Since the domain of \eqref{eq:op_FIE} is not necessarily  compact, the minimum of \eqref{eq:op_FIE}  may not be attained. Nevertheless, this does not affect the existence of a robust UISE.  Let $J^{\text{inf}}_k$ be the infimum of \eqref{eq:op_FIE} and  $\boldsymbol{\xi}_k:=(\hat{\boldsymbol{x}}_{\cdot|k}, \hat{\boldsymbol{v}}_{\cdot|k},  \hat{\boldsymbol{w}}_{\cdot|k}, \hat{\boldsymbol{y}}_{\cdot|k})$ be a feasible solution of \eqref{eq:op_FIE} with $J_k(\boldsymbol{\xi}_k)$ being its cost. If $J^{\text{inf}}_k>0$, choose $\boldsymbol{\xi}_k$ such that $J_k( \boldsymbol{\xi}_k ) \leq a J^{\text{inf}}_k$ with some fixed $a > 1$; since $J^{\text{inf}}_k \leq J_k( x_0,  \boldsymbol{v}_{[0,k] }, \boldsymbol{y}_{[0, k] } )$, the proof of Lemma~\ref{lem:FIE} holds with an additional factor $(1+a)/2$.  If  $J^{\text{inf}}_k=0$, let us consider a sequence of feasible solutions $(\boldsymbol{\xi}^j_k)^\infty_{j=1}$  such that $J_k(\boldsymbol{\xi}^j_k) \to J^{\text{inf}}_k$ as $j \to \infty$. By applying \eqref{eq:strong_detect} and the weak triangle inequality, we get $\alpha(\|\hat{x}^n_{k|k}-\hat{x}^m_{k|k}\|) \leq J_k(\boldsymbol{\xi}^n_k)  + J_k(\boldsymbol{\xi}^m_k)$. Hence,  $  \alpha(\|\hat{x}^n_{k|k}-\hat{x}^m_{k|k}\|) \to  0$ as $m, n \to \infty$. This shows that  $(\hat{x}^j_{k|k})^\infty_{j=1}$  is a  Cauchy sequence in $\mathbb{X} \subseteq \mathbb{R}^n$, and hence its limit exists and lies in $\mathbb{X}$, which  is closed.   Passing this limit to $\hat{x}_k$ in the proof of Lemma~\ref{lem:FIE} yields the concluding inequality  with the factor $1/2$. 
}

\begin{thm}
\label{thm:detectability_necessary}
	There exists a \textcolor{black}{robust} UISE ($\Psi_k$) for the system~\eqref{eq:sys} if and only if the system~\eqref{eq:sys} is strongly nonlinearly detectable. 
\end{thm}

\begin{proof}
The ``if" statement follows directly from Lemma~\ref{lem:FIE}. 
The ``only if" statement follows reasoning similar to that in  \cite[Prop.~2.6]{allan2021}. Specifically, let us consider \textcolor{black}{two arbitrary trajectories $(\boldsymbol{x},  \boldsymbol{u}, \boldsymbol{w},  \boldsymbol{v}, \boldsymbol{y})$ and $(\widetilde{\boldsymbol{x}},  \boldsymbol{u}, \widetilde{\boldsymbol{w}},  \widetilde{\boldsymbol{v}}, \widetilde{\boldsymbol{y}})$ of \eqref{eq:sys} and  a} UISE 
\eqref{eq:Psi_k} with $\bar{x}_0=\widetilde{x}_0$, \textcolor{black}{$\bar{\boldsymbol{y}}_{[0,k]}= \widetilde{\boldsymbol{y}}_{[0,k]}$ and} $\bar{\boldsymbol{v}}_{[0,k]}= \widetilde{\boldsymbol{v}}_{[0,k]}$ for all $k\in \mathbb{I}_{[0,\infty)}$.  
\textcolor{black}{Applying \eqref{eq:RGES} in Definition~\ref{def:RGAS} to the trajectory $(\widetilde{\boldsymbol{x}},  \boldsymbol{u}, \widetilde{\boldsymbol{w}},  \widetilde{\boldsymbol{v}}, \widetilde{\boldsymbol{y}})$,} we can conclude that  $\hat{\boldsymbol{x}}_{ [0,k] } = \widetilde{\boldsymbol{x}}_{[0,k]}$ for all $k\in \mathbb{I}_{[0,\infty)}$. \textcolor{black}{Applying this and  \eqref{eq:RGES}  to the trajectory $(\boldsymbol{x},  \boldsymbol{u}, \boldsymbol{w},  \boldsymbol{v}, \boldsymbol{y})$} yields 
\begin{align*} 
\beta &(\|x_k-\widetilde{x}_k\|)  \leq  \beta_0( \|x_0-\widetilde{x}_0 \|, k )  +\sum_{i=0}^{k} \bigl( \beta_v (\|v_{k-i}-\widetilde{v}_{k-i}\| , i ) + \textcolor{black}{\beta_y} (\|y_{k-i}-\widetilde{y}_{k-i}\| , i )   \bigr), 
\end{align*}
which completes the proof.  
\end{proof}

 Theorem~\ref{thm:detectability_necessary} holds even when $\alpha_0$ in \eqref{eq:strong_detect} and $\beta_0$ in \eqref{eq:RGES} are in class $\mathcal{GKL}$, which follows from  the same arguments used in the proof of Theorem~\ref{thm:detectability_necessary}  along with  the fact that  the weak triangular  inequality  \cite[(6)]{jiang1994} applies   to any $\alpha_0(r,k)$ in  $\mathcal{GKL}$ for any fixed $k$ as indicated by Proposition~\ref{prop:weak_triangular} in the Appendix.

%
 The following result shows that, for a linear system described by
\begin{equation}
	\label{eq:UIO_LTI}
	x_{k+1}= Ax_k+ B w_k, \ y_k= Cx_k
\end{equation}
with $(x_k, w_k, y_k)\in \textcolor{black}{\mathbb{R}^{n} \times \mathbb{R}^{q} \times \mathbb{R}^{p}}$ \textcolor{black}{as the domain of its trajectory}
for $k\in \mathbb{I}_{[0, \infty)}$, 
  the proposed strong nonlinear detectability  reduces to the notion of strong$^\ast$ detectability \cite{hautus1983}, thereby generalizing strong$^\ast$ detectability to nonlinear systems in view of Theorem~\ref{thm:detectability_necessary}.    

\begin{lem}\label{lem:UIO_LTI}
The system~\eqref{eq:UIO_LTI} is strongly nonlinear detectable according to Definition~\ref{def:strong_detect} if and only if (i)   \blu{ $y_k=0$ for  $k\in \textcolor{black}{\mathbb{I}_{[1, \infty)}}$ implies $\lim_{k\to \infty}\|x_k\|= 0$},   and (ii)  $\text{rank}(CB)=\text{rank}(B)$.
\end{lem}
\begin{proof}
\textbf{Necessity}: The condition~(i) follows immediately from \eqref{eq:strong_detect}. To show the condition~(ii), we employ a proof by contradiction. Specifically,  let us assume the condition~(ii) does not hold, i.e., $\text{rank}(CB) \neq \text{rank}(B)$. As a result, we have $\text{rank}(CB)< \text{rank}(B)$, which  implies
$\text{ker}(B) \subset \text{ker}(CB)$, and consequently $ \text{ker}(CB) \setminus \text{ker}(B) \neq \emptyset $. 
Then we can choose  a trajectory $(\boldsymbol{x},  \boldsymbol{w}, \boldsymbol{y}) $ with  $x_0=0$, $w_0 \in \text{ker}(CB) \setminus \text{ker}(B)$, $w_i=0$ for $i \in \mathbb{I}_{[1, \infty)}$ and another trajectory  $(\widetilde{\boldsymbol{x}}, 
 \widetilde{\boldsymbol{w}}, \widetilde{\boldsymbol{y}})$ with $\widetilde{x}_0=0$ and $\widetilde{w}_i \equiv 0$ satisfying \eqref{eq:UIO_LTI}.  
It follows  that 
\noindent \begin{align*}
x_k-\widetilde{x}_k &=A(x_{k-1}-\widetilde{x}_{k-1})+B (w_{k-1}-\widetilde{w}_{k-1}), \\
y_{k}-\widetilde{y}_k&=CA(x_{k-1}-\widetilde{x}_{k-1})+ CB(w_{k-1}-\widetilde{w}_{k-1})	
\end{align*}
with $k\in \mathbb{I}_{[1, \infty)}$. Hence,  $x_1-\widetilde{x}_1 \neq 0$ and $y_1-\widetilde{y}_1=0$, which together with $y_0-\widetilde{y}_0=C(x_0-\widetilde{x}_0)=0$ implies that
\eqref{eq:strong_detect} does not hold for $k=1$, thereby showing that  the system~\eqref{eq:UIO_LTI} is not strong nonlinear detectable. Therefore,  $\text{rank}(CB) = \text{rank}(B)$. 

\textbf{Sufficiency}: Let us denote $l:=\text{rank}(B)$. Then there exists a nonsingular matrix $T \in \mathbb{R}^{q \times q}$, such that $BT=(\widetilde{B}, 0_{n \times (q-l)})$ and $\widetilde{B}$ has full column rank. The state dynamic of \eqref{eq:UIO_LTI} can be reformulated into
$$x_{k+1}=Ax_k+\widetilde{B} \widetilde{T} w_k $$
with $ \widetilde{T}:= (I_l, 0_{l \times(q-l)}) T^{-1}$. 
Consequently, we have
\begin{subequations}
	\noindent \begin{align}
		x_{k}-\widetilde{x}_k &=A(x_{k-1}-\widetilde{x}_{k-1})+\widetilde{B}\widetilde{T}(w_{k-1}-\widetilde{w}_{k-1}), \label{eq:proof_UIO_LTI_1} \\
	y_{k}-\widetilde{y}_k&=C (x_k-\widetilde{x}_k),\label{eq:proof_UIO_LTI_2}
	\end{align}
\end{subequations}
for any two  trajectories $(\boldsymbol{x}, \boldsymbol{w}, \boldsymbol{y})$,  $(\widetilde{\boldsymbol{x}}, \widetilde{\boldsymbol{w}}, \widetilde{\boldsymbol{y}})$  of \eqref{eq:UIO_LTI}. 
Due to (ii), we have $\text{rank}(\widetilde{B})=\text{rank}( C\widetilde{B})$, implying that $C\widetilde{B}$ admits a left inverse $Y$. Hence, inserting \eqref{eq:proof_UIO_LTI_1} into \eqref{eq:proof_UIO_LTI_2}  and then solving for $\widetilde{T}(w_{k-1}-\widetilde{w}_{k-1})$ gives
$$
\widetilde{T}(w_{k-1}-\widetilde{w}_{k-1})= Y(y_k -\widetilde{y}_k) -YCA(x_{k-1}-\widetilde{x}_{k-1}). 
$$ 
Inserting into \eqref{eq:proof_UIO_LTI_1} yields
\begin{equation}
	\label{eq:proof_UIO_LTI_3}
	x_{k}-\widetilde{x}_k = \textcolor{black}{\widetilde{A}}(x_{k-1}-\widetilde{x}_{k-1})+\widetilde{B}Y(y_k-\widetilde{y}_k),
\end{equation}
\textcolor{black}{with $\widetilde{A}:=(A-\widetilde{B}YCA)$.} Choosing $\widetilde{x}_0=0$ as well as  $\widetilde{w}_k \equiv 0$ for all $k\in \mathbb{I}_{[0, \infty)}$, \textcolor{black}{
	we get  $\widetilde{x}_k \equiv 0$ and  $\widetilde{y}_k \equiv 0$ according to  \eqref{eq:UIO_LTI}. Hence, given any $(\boldsymbol{x},  \boldsymbol{y})$ satisfying}
$$ 
\textcolor{black}{x_{k} =  \widetilde{A}x_{k-1}+\widetilde{B}Y y_k, \ y_k=Cx_k,  }
$$
\textcolor{black}{since 
	$
	\widetilde{T}w_{k-1}= Yy_k-YCA x_{k-1}
	$
	 and $\widetilde{T}$ is surjective, there exists $\boldsymbol{w}$ such that $(\boldsymbol{x}, \boldsymbol{w}, \boldsymbol{y})$ is the trajectory of \eqref{eq:UIO_LTI}.   Consider any $\xi$ satisfying $C\widetilde{A}^k \xi=0$ for all $k\in \mathbb{I}_{[0, \infty)}$ and define $\xi_k:=\widetilde{A}^k \xi$, then $\boldsymbol{\xi}$ is the state trajectory of \eqref{eq:UIO_LTI} with  $y_k=C \xi_k = 0$ for all $k$ and some $\boldsymbol{w}$. Thus,  the condition~(i) yields  $\lim_{k\to \infty}\xi_k=0$. Consequently,  $(C, \widetilde{A})$  is detectable and  there exists a matrix $L$ such that $A_L:=\widetilde{A} + LC$ is Schur stable. Combining \eqref{eq:proof_UIO_LTI_3} and \eqref{eq:proof_UIO_LTI_2} gives } 
\begin{equation*}
	\delta x_{k} = \textcolor{black}{A_L}\delta x_{k-1}-L\delta y_{k-1}+\widetilde{B}Y \delta y_k  
\end{equation*}
with $\delta x_k := x_k-\widetilde{x}_k$ and $\delta y_k := y_k - \widetilde{y}_k$.  
Bounding its closed-form solution  shows the strong nonlinear detectability of \eqref{eq:UIO_LTI}.
\end{proof}

\subsection{MHE-based \textcolor{black}{robust} UISE} \label{sec:MHE_full}

As the complexity of the FIE increases over time, we consider estimating $x_k$ at each time $k\in \mathbb{I}_{[0, \infty)}$ via MHE using  the batch of input signals $\boldsymbol{u}_{[k-N_k, k-1]}$, \textcolor{black}{estimated outputs  $\boldsymbol{\bar{y}}_{[k-N_k, k]}$ and estimated noise $\boldsymbol{\bar{v}}_{[k-N_k, k]}$} with  $N_k = \min(N, k)$, where $N \in \mathbb{N}$ is the prediction horizon. 
At each time $k$,   the estimate $\hat{x}^\star_{k|k}$ for $x_k$  is computed by minimizing $J_{N_k}(\hat{x}_{k-N_k |k },  \hat{\boldsymbol{v}}_{\cdot |k }, \hat{\boldsymbol{y}}_{\cdot |k })$ 
subject to the constraint \eqref{eq:op_FIE_input}, $j\in \mathbb{I}_{[k-N_k, k-1]}$,  and constraints \eqref{eq:op_FIE_output}-\eqref{eq:op_FIE_output_2} with $j\in \mathbb{I}_{[k-N_k, k]}$.  The cost function $J_{N_k}$ is defined by
\begin{equation}
\label{eq:cost_MHE_full}
\begin{aligned}
J_{N_k}(\hat{x}_{k-N_k|k },   \hat{\boldsymbol{v}}_{\cdot |k }, \hat{\boldsymbol{y}}_{\cdot |k }) &=  \mu^{N_k}\alpha_2(2\| \hat{x}_{k-N_k|k}-\hat{x}_{k-N_k} \|) + \sum_{j=0}^{N_k}\mu^{j-1}2  \sigma_v( 2\| \hat{v}_{k-j|k}-\bar{v}_{k-j}\|) \\
&+\sum_{j=0}^{N_k}\mu^{j-1}2  \sigma_y(\blue{2\|\bar{y}_{k-j}-\hat{y}_{k-j|k}\|}), 
\end{aligned}
\end{equation} 
with \blue{$\mu$, $\alpha_2$, $\sigma_v$, $\sigma_y$} chosen according to \eqref{eq:Lyapunov} in Theorem~\ref{theo:Lyapunov}, provided that the system~\eqref{eq:sys} is strongly nonlinearly detectable. The prior $\hat{x}_{k-N_k}$ is determined as follows
\begin{equation}
\label{eq:prior_rule}
	\hat{x}_{k} = \begin{cases}
		\hat{x}^\star_{k | k} &  k \in  \mathbb{I}_{[1, \infty)}\\
		\bar{x}_0 & k=0. 
	\end{cases}
\end{equation}
Let us refer to the resulting MHE  as $ \mathrm{MHE_{N_k}}$. In the following, by adapting the Lyapunov-like technique from  \cite{schiller2023},  
we  show that the $\mathrm{MHE_{N_k}}$ with a sufficiently long  horizon $N$ constitutes a \textcolor{black}{robust} UISE.
\begin{thm}\label{theo:MHE_full}
	Suppose that the system~\eqref{eq:sys} is strongly nonlinearly detectable.  Let the horizon $N$ be chosen such that
	\begin{equation}
	 \label{eq:MHE_full_horizon}  (2\mu^{N}\alpha_2 \circ  2 \alpha^{-1}_1)(r)< \rho r, \ \forall  r \in (0, \infty)
	\end{equation}
	with  $\mu, \alpha_1, \alpha_2$ specified in \eqref{eq:Lyapunov} and some $\rho \in(0,1)$, 
then the estimator $\mathrm{MHE_{N_k}}$ is a \textcolor{black}{robust} UISE. 
\end{thm}
\begin{proof}
	In essence, the proof resembles that of \cite[Corollary~1]{schiller2023}, but proceeds with more general  Lyapunov-like functions and supply rates than the ones used in  \cite[Corollary~1]{schiller2023}. Moreover,  because \blue{current estimates $\bar{y}_k$ and $\bar{v}_k$} appear in the cost~\eqref{eq:cost_MHE_full} and  \eqref{eq:Lyapunov_2},   additional adjustments are necessary.
 Following from  strong nonlinear detectability,  we can 
invoke \eqref{eq:Lyapunov_2} and consequently its generalization \eqref{eq:lyapunov_proof_1} in the proof of Theorem~\ref{theo:Lyapunov}  to get  
\begin{align*}
	&V(x_k,\hat{x}^\star_{k|k})
	\leq \mu^{N}V(x_{k-N}, \hat{x}^{\star}_{k-N|k})+ 
	   \sum^{N}_{j=0} 2\mu^{j-1}\bigl( \sigma_v( \|v_{k-j}- \hat{v}^\star_{k-j|k} \|)  + \sigma_y (\| y_{k-j} -\hat{y}^\star_{k-j|k} \|) \bigr) \\
	& \overset{\eqref{eq:Lyapunov_1}}{\leq}\sum^{N}_{j=0} 2\mu^{j-1}\bigl( \sigma_v(2 \| \hat{v}^\star_{k-j|k}-\bar{v}_{k-j} \|) +  \sigma_v (2\| \delta v_{k-j}\|)\big) 
	   + \sum^{N}_{j=0} 2\mu^{j-1} \sigma_y(\| y_{k-j} -\hat{y}^\star_{k-j|k} \|)  \\
	&  + \mu^{N}\alpha_2(\| \textcolor{black}{x_{k-N}-\hat{x}_{k-N}+\hat{x}_{k-N} -  \hat{x}^\star_{k-N|k}}\|) \\
	&\overset{\eqref{eq:prior_rule}}{\leq} \mu^{N}   \alpha_2( 2\|x_{k-N} - \hat{x}^\star_{k-N|k-N} \|)+
	  \sum^{N}_{j=0} \mu^{j-1} 2\bigl( \sigma_v( 2\|\delta v_{k-j} \|) \blue{+ \sigma_y( 2\|\delta y_{k-j} \|)} \bigr)  
	  +  J_{N}(\hat{x}^{\star}_{k-N|k}, \hat{\boldsymbol{v}}^\star_{\cdot|k}, \hat{\boldsymbol{y}}^\star_{\cdot|k}), 
\end{align*}
for  $k\in \mathbb{I}_{(N, \infty)}$, where $\delta v_{k-j}:= v_{k-j} -\bar{v}_{k-j}$ \blue{and $\delta y_{k-j}:= y_{k-j} -\bar{y}_{k-j}$ }. 
By optimality, the above inequality implies
\begin{equation}\label{eq:fullMHE_proof}
\begin{aligned}
   V( x_k, \hat{x}^\star_{k|k})  &\leq 	 \sum^{N}_{j=0} \mu^{j-1} 4 \left( \sigma_v( 2\|\delta v_{k-j} \|) \blue{+\sigma_y( 2\|\delta y_{k-j} \|)    } \right) 
	  +\blu{2}\mu^{N}   \alpha_2( 2\|x_{k-N} - \hat{x}^\star_{k-N|k-N} \|)\\
	& \overset{\eqref{eq:Lyapunov_1}}{\leq} (2\mu^{N}\alpha_2 \circ  2 \alpha^{-1}_1)\bigl( V( x_{k-N},  \hat{x}^\star_{k-N|k-N} )\bigr)
	 +\sum^{N}_{j=0} \mu^{j-1} 4 \bigl( \sigma_v( 2\|\delta v_{k-j} \|) \blue{+\sigma_y( 2\|\delta y_{k-j} \|)    } \bigr)  \\  &\overset{\eqref{eq:MHE_full_horizon}}{\leq} \sum^{N}_{j=0} \mu^{j-1} 4 \bigl( \sigma_v( 2\|\delta v_{k-j} \|) \blue{+\sigma_y( 2\|\delta y_{k-j} \|)    } \bigr) 
	  	+ \rho V( x_{k-N},  \hat{x}^\star_{k-N|k-N} ). 
\end{aligned}
\end{equation}
Similarly, for $k \in \mathbb{I}_{[0, N]}$,  we have
\begin{equation}
	\label{eq:fullMHE_proof_1}
	\begin{aligned}
		V(x_k,\hat{x}^\star_{k|k}) &\leq   \mu^{k-N} \rho V( x_{0} , \bar{x}_0 ) + \sum^{\textcolor{black}{k}}_{j=0} \mu^{j-1} 4 \bigl( \sigma_v( 2\|\delta v_{k-j} \|) \blue{+\sigma_y( 2\|\delta y_{k-j} \|)    } \bigr) .
	\end{aligned}
\end{equation}
Let $ \tau:= k- \lfloor k/N \rfloor N$ for $k\in \mathbb{I}_{[0,\infty)}$ and  $\lambda := \max(\rho^{1/N}, \mu) \in (0,1)$.  If $k > N$, invoking \eqref{eq:fullMHE_proof}  $ \lfloor k/N \rfloor$ times together with \eqref{eq:fullMHE_proof_1}  leads to 
\begin{align*}
V(x_k, \hat{x}^\star_{k|k})  & \leq  
\rho^{\lfloor \frac{k}{N} \rfloor}   \mu^{\tau-N} \rho V ( x_0, \bar{x}_0)  
 + \rho^{\lfloor \frac{k}{N} \rfloor}  \sum^{\tau}_{j=0}  \mu^{j-1} 4 \bigl(\sigma_v ( 2 \|\delta v_{\tau-j}\|) \blue{+ \sigma_y ( 2 \|\delta y_{\tau-j}\|)} \bigr) \\
& +\sum^{ \lfloor \frac{k}{N} \rfloor -1}_{i=0} \rho^{i}  \sum^{N}_{j=0} \mu^{j-1} 4 \big(\sigma_v (2 \|\delta v_{k-iN-j}\|) \blue{+ \sigma_y ( 2 \|\delta y_{k-iN-j}\|)} \big). 
\end{align*}
As $\rho^{1/N}< \lambda$, the above condition implies
\begin{align*}
 V(x_k, \hat{x}^\star_{k|k})  &\leq   \sum^{ \lfloor \frac{k}{N}\rfloor -1}_{i=0} \lambda^{iN}  \sum^{N}_{j=0} \lambda^{j-1} 4 \bigl(\sigma_v ( 2\|\delta v_{k-iN-j}\|) \blue{+\sigma_y ( 2\|\delta y_{k-iN-j}\|) } \bigr)  \\
 &   + \lambda^{\lfloor \frac{k}{N} \rfloor N}  \bigl( \frac{\rho \lambda^\tau}{\mu^N}  V (x_0,\bar{x}_0) + \sum^{\tau}_{j=0}  \lambda^{j-1} 4 \big(\sigma_v (2\|\delta v_{\tau-j}\|) +\blue{\sigma_y (2\|\delta y_{\tau-j}\|)}  \big)\bigr)\\
& \leq   \frac{\rho \lambda^k}{\mu^N} V(  x_0, \bar{x}_0) +  \lambda^{\lfloor \frac{k}{N} \rfloor N} \sum^{\tau}_{j=1}  \lambda^{j-1} 8 \big(\sigma_v(2\|\delta v_{\tau-j}\|) \blue{+\sigma_y(2\|\delta y_{\tau-j}\|)} \big) \\
& + \sum^{ \lfloor \frac{k}{N} \rfloor -1}_{i=0} \lambda^{iN}  \sum^{N}_{j=1} \lambda^{j-1} 8 \big( \sigma_v ( 2\|\delta v_{k-iN-j}\|) \blue{+  \sigma_y ( 2\|\delta y_{k-iN-j}\|)  } \big)  + \frac{4}{\lambda} \big( \sigma_v (2\|\delta v_k \|) \blue{ + \sigma_y (2\|\delta y_k \|)}  \big)  \\ 
&  \leq  \frac{\rho \lambda^k}{\mu^N}  V(  x_0, \bar{x}_0) + \sum^{k}_{i=0} 8 \lambda^{i-1}  \big(\sigma_v ( 2\| \delta  v_{k-i}\|) \blue{+ \sigma_y ( 2\| \delta  y_{k-i}\|) }  \big). 	
\end{align*}
In case of  $k \in \mathbb{I}_{[0,N]}$,  the above inequality  follows readily from \eqref{eq:fullMHE_proof_1}. Therefore,   the above condition holds for all $k \in \mathbb{I}_{[0, \infty)}$.  By applying the lower and upper bounds of $V$ in \eqref{eq:Lyapunov_1} to this condition and noting  that $\lambda \in (0,1)$,  we can conclude that  $\mathrm{MHE_{N_k}}$ is \textcolor{black}{a robust} UISE according to Definition~\ref{def:RGAS}.  
\end{proof}

\begin{rem}
\label{rem:finite_N} There may not exist  $\rho \in(0,1)$ and a finite $N$ such that \eqref{eq:MHE_full_horizon} holds. Indeed, let us consider $\alpha_1(r)= r^2$ for $r\in [0, \infty)$,  $\alpha_2(r)=b r^4$ for $r\in(1, \infty)$ and $\alpha_2(r)=b r^2$ for $r\in [0,1]$ with $b>1$. Then  $(2\mu^N\alpha_2 \circ 2\alpha^{-1}_1)(r)=32\mu^N br^2 > r$ for all $r> \max((32\mu^N b)^{-1}, 1)$. \blu{To ensure the existence of  $N$, one can restrict $\alpha_1$ and $\alpha_2$  by $c_1 r^n \leq \alpha_1(r)  \leq \alpha_2(r) \leq c_2 r^n$ with some $0< c_1 < c_2$ and $n >0$.  }  
\end{rem}


\section{Unknown input state estimators with residual models} \label{sec:UISE_red}

\subsection{A tighter notion for strong nonlinear detectability}
 \blue{To enhance the computational efficiency of the  MHE-based UISE,   we will employ  residual models and past measured outputs for the MHE formulation.} This necessitates linking strong nonlinear detectability  with detectability (or i-IOSS) of   residual models. To this end, we specify additional system properties that enable us to infer strong nonlinear detectability from  detectability of  residual models. 
Specifically,  we first assume the following:
\begin{assum}\label{ass:compact}
The domains $\mathbb{X}$, $\mathbb{V}$ and $\mathbb{Y}$  for the state, noise and  output of the system~\eqref{eq:sys} are compact. 
\end{assum}
Next, let us consider  a diffeomorphism   $T(x)=\text{col}\big(T_\flat(x), T_\sharp(x)\big) \in \mathbb{R}^{n_\flat + n_\sharp}$ with $x \in \mathbb{R}^n$.     Performing  a coordinate transformation  on \eqref{eq:sys} via $z_k:=T(x_k)=\text{col}(z_k^\flat, z_k^\sharp)$ results in 
\begin{subequations}
	\label{eq:sys_new}
	\noindent \begin{align}
	z^\flat_{k+1} &=  f_\flat(z^\sharp_k, z^\flat_k,u_k,w_k), \label{eq:sys_new_1}\\ 
	z^\sharp_{k+1} &=  f_\sharp(z^\sharp_{k}, z^\flat_{k},  u_k,w_k),    \label{eq:sys_new_2} \\
	y_k &= h_T(z^\sharp_k, z^\flat_k,   v_k),   \label{eq:sys_new_3}	
	\end{align}
\end{subequations}
	where the functions $f_\flat$, $f_\sharp$ and $h_T$ are defined by
	\begin{align*}
    &f_{\diamond}(z^\sharp_k, z^\flat_k,u_k,w_k):= T_{\diamond}\Bigl( f\big( T^{-1}(z_k),u_k, w_k\big)\Bigr), \ \diamond \in \{\flat,\sharp\}\\
    &  h_T(z^\sharp_k, z^\flat_k,   v_k):= h\big(T^{-1}(z_k),  v_k \big). 
	\end{align*}
For this transformed system~\eqref{eq:sys_new}, which is equivalent to \eqref{eq:sys}, let us define a set $\overline{\mathbb{X}}:= \{ T(x) \in \mathbb{R}^n | x\in \mathbb{X} \}$. Due to the diffeomorphic property of $T$ and the compactness of $\mathbb{X}$, $x\in \mathbb{X} \Leftrightarrow T(x)\in \overline{\mathbb{X}}$ and $\overline{\mathbb{X}}$ is compact.   Let $\overline{\mathbb{X}}_\diamond$ denote the projection of $\overline{\mathbb{X}}$ onto $z_\diamond$ for $\diamond=\{\flat, \sharp\}$, then $\overline{\mathbb{X}}  \subset \overline{\mathbb{X}}_\flat \times \overline{\mathbb{X}}_\sharp$, and both $\overline{\mathbb{X}}_\flat$ and $\overline{\mathbb{X}}_\sharp$ are compact.  We then introduce the following assumptions on \eqref{eq:sys_new}.   
\begin{assum}
\label{ass}
There exists a diffeomorphism $T=(T_\flat, T_\sharp): \mathbb{R}^n \rightarrow \mathbb{R}^{n_\flat+n_\sharp}$  such that   a continuously differentiable function $\psi: \mathbb{R}^{p}  \times \mathbb{R}^{n_\sharp} \times \mathbb{R}^{r} \rightarrow \mathbb{R}^{n_\flat}$ satisfying 
\begin{equation}
 \label{eq:ass_1}
z^\flat = \psi(y,z^\sharp, v) 
\end{equation}
for any solution $(z^\flat, z^\sharp, y,  v)$ of \eqref{eq:sys_new_3}  exists,  and  
\begin{equation}
\label{eq:ass_2}
\frac{\partial f^i_\sharp}{\partial w}(z^\sharp, z^\flat, u, w) = 0, \ \forall i\in \mathbb{I}_{[1, n_\sharp]} 
\end{equation}
with $f^i_\sharp$ denoting the $i$-th element of  $f_\sharp$ from \eqref{eq:sys_new_2}.
\end{assum}

\begin{rem}
If  $n_\flat > p$, then there exists no  $\psi$ satisfying \eqref{eq:ass_1}, since $h_T(z^\sharp, \cdot,  v)$ in \eqref{eq:sys_new_3} is not injective for any $(z^\sharp, v)$ by the Borsuk-Ulam theorem \cite{matouvsek2003}. 
Furthermore, by the global implicit function theorem  \cite[Theorem~1]{zhang2006},   the existence of  a continuously differentiable function $ \psi$ satisfying \eqref{eq:ass_1} is ensured if there exist distinct indices $i_1, \ldots, i_{n_\flat} \in \mathbb{I}_{[1, p]} $ and a fixed $d>0$  such that 
$$
\left\|  \frac{\partial h^{i}_T}{\partial z^{\flat, i}}(z^\sharp,  z^\flat,   v) \right\|-\sum_{j\neq i} \left\| \frac{\partial h^{i}_T}{\partial z^{\flat, j}}(z^\sharp,  z^\flat,  v)  \right\| \geq d 
$$
for all $(z^\sharp,  z^\flat,  v) \in \mathbb{R}^{n_\sharp} \times \mathbb{R}^{n_\flat} \times \mathbb{R}^{r}$ and $ i\in \{ i_1, \ldots, i_{n_\flat} \}$, where $h^{i}_{T}$ and $z^{\flat, i}$ denote the $i$-th element of the function $h_T$ and the  vector $z^\flat $ respectively.  
\end{rem}

 For a class of nonlinear systems,  there is a systematic way to find $T$ satisfying Assumption~\ref{ass}. 
Specifically, consider a nonlinear system \eqref{eq:sys}, which can be transformed by a diffeomorphism $\bar{x}=\phi(x)$ into the following quasi-state-affine systems
\begin{subequations}
	\label{eq:sys_spec}
	\noindent \begin{align}
	\bar{x}_{k+1}& =f_1(\bar{x}_k,u_k)+B f_2(w_k, \bar{x}_k, u_k),   \label{eq:sys_spec_1} \\
	y_k&=C \bar{x}_k + g( v_k),    \label{eq:sys_spec_2}
	\end{align}
\end{subequations}
where $C\in \mathbb{R}^{p \times n}, B\in \mathbb{R}^{n \times s}$, and  $f_1, f_2, g$ are continuously differentiable.  

\begin{prop}
\label{prop:special_case} \blu{For the system~\eqref{eq:sys_spec}, if}
 $\text{rank}(CB)=\text{rank}(B)$, then there exists  $T$ satisfying Assumption~\ref{ass}. 
\end{prop}
\begin{proof}
Let us denote  $\text{rank}(CB)$ by $r$. The condition $\text{rank}(CB)=\text{rank}(B)$ implies  $n, p, s \geq r$. Furthermore,  there exists a nonsingular matrix $S\in \mathbb{R}^{p \times p}$ and a full row rank matrix $Q \in  \mathbb{R}^{r \times s}$  such that $SCB=\text{col}( Q,  \boldsymbol{0}_{(p-r) \times s })$.  Partitioning $S$ into $S=\text{col}( S_\flat, S_\sharp)$ with $S_\flat\in \mathbb{R}^{r \times p}$, we get $S_\flat CB=Q$. This along with $\text{rank}(B)=\text{rank}(Q)$ implies that there exists a full row rank matrix $M\in \mathbb{R}^{(n-r)\times n}$ such that $ \text{ker}(M) = \text{im}(B)  \subseteq \mathbb{R}^{ n }$ and  $\text{col}( S_\flat C, M )$ is nonsingular. Choosing 
\begin{equation}
	\label{eq:choice_special}
	\begin{aligned}
     &T_\flat(x)= S_\flat C \phi(x), \  T_\sharp(x)= M\phi(x),  \\
     & \psi(y, z^\sharp,  v) =  S_\flat (y-g(v)),
	\end{aligned}
\end{equation}
we can verify  that  $T=(T_\flat, T_\sharp)$  is a diffeomorphism and that conditions in  Assumption~\ref{ass} are satisfied.  
\end{proof}
Finding $\phi$, such that  output function is affine in $\phi(x)$ as in \eqref{eq:sys_spec_2},  is possible for many physical systems; see, e.g.  \cite[Section 5]{nijmeijer1990}. Moreover,  the state dynamic~\eqref{eq:sys_spec_1} is general in its form.  Therefore, the restrictions posed by  Assumption~\ref{ass} are not severe and are suited for practical applications. 

Under  Assumption~\ref{ass}, the subsystem~\eqref{eq:sys_new_2}--\eqref{eq:sys_new_3} can be  decoupled from \eqref{eq:sys_new_1} and  transformed to the following residual system, 
\begin{equation}
\label{eq:sys_red}
\noindent \begin{aligned}
 z^\sharp_{k+1} & = \hat{f}_\sharp(z^\sharp_k, \gamma_k, u_k,  v_k):= f_\sharp\big(  z^\sharp_k, \psi(\gamma_k,  z^\sharp_k,  v_k), u_k, w_k \big),   \\
y_k &= \hat{h}_T(z^\sharp_k,  \gamma_k,  v_k):= h_T\big( z^\sharp_k, \psi(\gamma_k,  z^\sharp_k, v_k),  v_k \big), 
\end{aligned}
\end{equation}
with  $\psi$ satisfying \eqref{eq:ass_1} and the fictitious input $\gamma_k \in \mathbb{R}^p$ arsing from the decoupling.  Notably, due to the decoupling, the state and output of the residual system \eqref{eq:sys_red} are not linked to those of the original system~\eqref{eq:sys} or the transformed system~\eqref{eq:sys_new}.   Yet, with a slight abuse of notation,  we adopt  $z^{\sharp}_k$ and $y_k$ from the  transformed full system~\eqref{eq:sys_new} to denote the state and output of  the decoupled subsystem \eqref{eq:sys_red} to ease notation. 
Moreover,    it cannot be concluded from  $y = h_T ( z^\sharp,  \psi( \gamma, z^\sharp, v),  v )$  that $y=\gamma$, since $h_T(z^\sharp, \cdot , v ): \mathbb{R}^{n_\flat} \to \mathbb{R}^p$ is not necessarily the left inverse function of $\psi ( \cdot,  z^\sharp, v ): \mathbb{R}^p \to \mathbb{R}^{n_\flat}$ for all $(z^\sharp, v)$.

In the following, we show that  strong nonlinear detectability from Definition~\ref{def:strong_detect} can be inferred from  detectability (or i-IOSS)  of the  residual system~\eqref{eq:sys_red} under assumptions.

\begin{lem} \label{lem:p-eioss}
		Suppose \blu{Assumption~\ref{ass:compact}  holds.}  
	The system \eqref{eq:sys} is strongly nonlinearly detectable  if  the residual system~\eqref{eq:sys_red} induced from system \eqref{eq:sys}  via the map $T$, which  satisfies Assumption~\ref{ass},  is i-IOSS, i.e.,  there exist $\beta \in \mathcal{K}_{\infty}$, $\beta_0 \in \mathcal{KL}$ and summable $\mathcal{KL}$-functions $\beta_v$, $\beta_y$ and $\beta_\gamma$ such that 
	\begin{equation}
		\label{eq:p-eioss}
		\begin{aligned}
			\beta(\|z^\sharp_k-\widetilde{z}^\sharp_k\|) &\leq   \sum^{k}_{i=1} \beta_v ( \| v_{k-i} - \widetilde{v}_{k-i}\|, i) + \sum^{k}_{i=1} \big(\beta_y (\|y_{k-i}- \widetilde{y}_{k-i} \|, i) +\beta_\gamma (\|\gamma_{k-i} -\widetilde{\gamma}_{k-i} \|, i ) \big) \\
			& + \beta_0 ( \| z^\sharp_0 -\widetilde{z}^\sharp_0\|, k)
		\end{aligned}
	\end{equation}
	for any $(\boldsymbol{z}^\sharp, \boldsymbol{\gamma}, \boldsymbol{u}, \boldsymbol{v},  \boldsymbol{y}), (\widetilde{\boldsymbol{z}}^\sharp, \widetilde{\boldsymbol{\gamma}},  \boldsymbol{u},  \widetilde{\boldsymbol{v}},  \widetilde{\boldsymbol{y}}) \in \mathbb{Z}_{\sharp}^\infty$ satisfying \eqref{eq:sys_red} with   $ \mathbb{Z}_{\sharp}:=\overline{\mathbb{X}}_{\sharp} \times \mathbb{Y}  \times \mathbb{U} \times \mathbb{V}  \times \mathbb{Y} $. 
\end{lem}

\begin{proof}
		In the proof, let us abbreviate $\diamond_k - \widetilde{\diamond}_k$ by $\delta \diamond_k$ for $k\in \mathbb{I}_{[0, \infty)}$ and $\diamond \in \{z^\sharp, v, y \}$.  
	Consider solutions $(\boldsymbol{x}, \boldsymbol{u}, \boldsymbol{w}, \boldsymbol{v}, \boldsymbol{y} )$,  $(\widetilde{\boldsymbol{x}}, \boldsymbol{u}, \widetilde{\boldsymbol{w}}, \widetilde{\boldsymbol{v}}, \widetilde{\boldsymbol{y}} )\in \prod_{i\in \mathbb{I}_{[0, \infty)}} \mathbb{Z}$ of \eqref{eq:sys}. Let $ z^\sharp_k = T_{\sharp}(x_k)$ and   $ \widetilde{z}^\sharp_k = T_{\sharp}(\widetilde{x}_k)$, then  $(\boldsymbol{z}^\sharp, \boldsymbol{y}, \boldsymbol{u}, \boldsymbol{v},  \boldsymbol{y}), (\widetilde{\boldsymbol{z}}^\sharp, \widetilde{\boldsymbol{y}},  \boldsymbol{u},  \widetilde{\boldsymbol{v}},  \widetilde{\boldsymbol{y}}) \in \mathbb{Z}^\infty_{\sharp}$ are the solution of \eqref{eq:sys_red}. Hence, by invoking \eqref{eq:p-eioss} and $\|T_\sharp(x) \| \leq \| T(x)  \|$ and defining $\bar{\beta}_{y}:= \beta_y + \beta_\gamma$, we get 
	\begin{equation}
		\label{eq:p-eioss_proof_1} 
		\begin{aligned}
			\beta (\|z^\sharp_k-\widetilde{z}^\sharp_k\|) \leq  \beta_0 ( \| T(x_0) -T(\widetilde{x}_0)\|, k)  
			 + \sum^{k}_{i=1} \big( \beta_v (\| \delta v_{k-i} \|, i) + \bar{\beta}_y (\| \delta y_{k-i} \|, i ) \big).
		\end{aligned} 
	\end{equation}
	Since $\mathbb{Y} \times \overline{\mathbb{X}}_{\sharp} \times \mathbb{V}$ is compact in view of  Assumption~\ref{ass:compact}, the continuously differentiable function~$\psi$ is Lipschitz continuous. Let us denote its Lipschitz constant by $L_{\psi}$.   
	From the condition \eqref{eq:ass_1}  and by recalling that $z_k^\flat = T_\flat(x_k)$, we arrive at 
	\begin{equation}
		\label{eq:p-eioss_proof_2}  
		\|T_\flat(x_k) - T_\flat(\widetilde{x}_k )\|\leq L_{\psi} (\|\delta y_k \|+\|\delta z^\sharp_k \| + \|\delta v_k\|).
	\end{equation}
	Let us define $\sigma_x(r):= \beta( (3+3L_\psi)^{-1} r) \in \mathcal{K}_{\infty}$ for $r\in [0, \infty)$.
	By exploiting this along with  \eqref{eq:p-eioss_proof_2} and the weak triangular inequality \cite[Equation.~6]{jiang1994},  we get 	
	\begin{equation}
		\label{eq:p-eioss_proof_3}
		\begin{aligned}
			\sigma_x( \|T(x_k)-T(\widetilde{x}_k)\|) & \leq \sigma_x \big(3(1+L_\psi) \| \delta z^\sharp_k\|\big) 
			+ \sigma_x (3L_\psi\|\delta y_k \| ) + \sigma_x (3L_\psi\|\delta v_k\|) \\ 
			& \overset{\eqref{eq:p-eioss_proof_1}}{\leq} \sum^{k}_{i=0} \sigma_v ( \| \delta v_{k-i}\|, i) + \sigma_y ( \|\delta y_{k-i} \|, i) 
			+ \beta_0 (\| T(x_0) -T(\widetilde{x}_0)\|, k),  
		\end{aligned}
	\end{equation}
	where  $\sigma_v(r,0):= \max(\beta_v(r, 0), \sigma_x(3L_\psi r) )$,  $\sigma_y(r,0):= \max(\bar{\beta}_y(r, 0), \sigma_x(3L_\psi r) )$  and $\sigma_v(r,i):= \beta_v(r,i)$ as well as $\sigma_v(r,i):= \bar{\beta}_y(r,i)$ for $i\in \mathbb{I}_{(0, \infty)}$, $r\in [0, \infty)$. Note that  $\sigma_v$, $\sigma_y$ are both summable $\mathcal{KL}$-functions. Finally, let $L_T$ and $L_{T^{-1}}$ be the Lipschitz constants of the diffeomorphism $T$ and its inverse on the compact sets $\mathbb{X}$ and $\overline{\mathbb{X}}$ respectively. Applying $L_{T^{-1}}^{-1}\|x-\widetilde{x}\|\leq \|T(x)-T(\widetilde{x})\| \leq L_T\|x-\widetilde{x}\|$ to \eqref{eq:p-eioss_proof_3} leads to  \eqref{eq:strong_detect}.   
\end{proof}

\begin{rem}
		For the  linear time-invariant system~\eqref{eq:UIO_LTI}, Lemma~\ref{lem:p-eioss} also holds in the reverse direction \blu{without Assumption~\ref{ass:compact}}. To see this, we first note that  the conditions~(i) and (ii) in Lemma~\ref{lem:UIO_LTI} hold due to strong nonlinear detectability.   We then follow the same line of reasoning as in the proof of Proposition~\ref{prop:special_case} to show that the condition~(ii)  implies the existence of linear maps $T_\flat $, $T_\sharp$  and  $\psi$  satisfying the conditions in Assumption~\ref{ass}, \textcolor{black}{and they can be chosen according to \eqref{eq:choice_special} with $\phi(x)=x$.}  Hence,  the transformed system~\eqref{eq:sys_new} and the residual system~\eqref{eq:sys_red} are linear. Moreover,  given any  $(\boldsymbol{z}^\sharp, \boldsymbol{\gamma} ,\boldsymbol{y})$ satisfying \eqref{eq:sys_red} with  $\gamma_k=y_k$ for $k\in \mathbb{I}_{[0, \infty)}$, due to  \eqref{eq:ass_1} and  the surjectivity of  $T_\flat (Bw)$, there exist $\boldsymbol{z}^\flat$ and $\boldsymbol{w}$  such that   $(\boldsymbol{z}^\flat, \boldsymbol{z}^\sharp,\boldsymbol{w}, \boldsymbol{y})$ satisfies \eqref{eq:sys_new}, and hence  $(\boldsymbol{x},\boldsymbol{w}, \boldsymbol{y})$ with $x_k=T^{-1}(z_k)$ is the trajectory of  the  system~\eqref{eq:UIO_LTI}.     \textcolor{black}{Let  $\bar{A}$ and $\bar{C}$ be the state and output matrix of the system~\eqref{eq:sys_red} respectively. Consider any $\xi$ satisfying   $\bar{C}\bar{A}^k \xi=0$ for all $k \in \mathbb{I}_{[1, \infty)}$ and define  $\xi_k:= \bar{A}^k \xi$, then $(\boldsymbol{\xi}, \boldsymbol{\gamma}, \boldsymbol{y})$ with $\gamma_k=y_k=\bar{C}\xi_k=0$ for all $k$ is the trajectory of \eqref{eq:sys_red}.  
	} Hence,   the condition~(i) and $ \lim_{x\to 0}T_\sharp(x)=0$  indicate that  \textcolor{black}{$(\bar{C}, \bar{A})$ is detectable}, thereby showing that the  system~\eqref{eq:sys_red} is i-IOSS in the sense of \eqref{eq:p-eioss}. 
\end{rem}

\subsection{Two-stage MHE-based UISE }
 Drawing upon the preconditions in  Lemma~\ref{lem:p-eioss} and the corresponding residual model~\eqref{eq:sys_red}, we construct a two-stage MHE-based UISE  that is  computationally \blue{more efficient than} the MHE-based \textcolor{black}{robust} UISE designed with full-order models. Specifically, the state estimation proceeds at each time  $k \in \mathbb{I}_{\geq 1}$ via two stages: first,  an MHE provides a one-step ahead estimate $\hat{z}^\sharp_{k}$ of the  state in the residual model~\eqref{eq:sys_red};  Second, \blue{the  current full-state estimate $\hat{x}_{k}$ is computed by} 
 \begin{equation}
 \label{eq:fullest}
	\textcolor{black}{	\hat{x}_k = T^{-1}\biggl( \mathrm{proj}_{\overline{\mathbb{X}}} \Bigl( \mathrm{col}\bigl(\psi(\blue{\bar{y}}_k,  \hat{z}^\sharp_k,  \bar{v}_k ), \hat{z}^\sharp_k \bigr) \Bigr)  \biggr)} 
 \end{equation} 
  \textcolor{black}{with an orthogonal projection   $\mathrm{proj}_{\overline{\mathbb{X}}}$ onto the set $\overline{\mathbb{X}}$, the current estimated output $\bar{y}_k$ and noise $\bar{v}_k$.}
 For the estimation of $z^\sharp_k$, we design an MHE considering past inputs $\boldsymbol{u}_{[k-N_k, k-1]}$, \blue{past  outputs $ \boldsymbol{y}_{[k-N_k, k-1]}$ with $\boldsymbol{\bar{y}}_{[k-N_k, k-1]}=\boldsymbol{y}_{[k-N_k, k-1]}$, }  priors $\bar{x}_0\in \mathbb{X}$,  $\bar{\boldsymbol{v}}_{[k-N_k, k-1]}$ as well as the past estimate $\hat{z}^\sharp_{k-N_k}$   within the horizon $N_k:= \min(k, N)$,  $N\in \mathbb{I}_{\geq 1}$. 
To ensure the existence of a finite stabilizing prediction horizon $N$, cf. Remark~\ref{rem:finite_N},  we impose the following stronger assumption on the residual system~\eqref{eq:sys_red}, requiring  exponential i-IOSS rather than  i-IOSS. 
\begin{assum}\label{ass:e-ioss}
The residual system~\eqref{eq:sys_red} derived through   the map $T$ from Assumption~\ref{ass} satisfies
\begin{equation}
	\label{eq:p-eioss_quadratic}
	\begin{aligned}
		 \|z^\sharp_k-\tilde{z}^\sharp_k\|^2   &\leq  \sum^{k-1}_{i=0} \mu^{k-i} ( c_v \| v_{i} - \tilde{v}_{i}\|^2 + c_y \|y_{i}- \tilde{y}_{i} \|^2) 
		+ \sum^{k-1}_{i=0} \mu^{k-i} c_\gamma \|\gamma_{i} -\tilde{\gamma}_{i} \|^2 +c_x\mu^k \| z^\sharp_0 -\tilde{z}^\sharp_0\|^2
	\end{aligned}
\end{equation}
for any trajectories $(\boldsymbol{z}^\sharp, \boldsymbol{\gamma}, \boldsymbol{u}, \boldsymbol{v},  \boldsymbol{y}), (\tilde{\boldsymbol{z}}^\sharp, \tilde{\boldsymbol{\gamma}},  \boldsymbol{u},  \tilde{\boldsymbol{v}},  \tilde{\boldsymbol{y}}) \in  \mathbb{Z}^{\infty}_{\sharp}$ of \eqref{eq:sys_red} with $ \mathbb{Z}_{\sharp}:=\overline{\mathbb{X}}_\sharp \times \mathbb{Y}  \times \mathbb{U} \times \mathbb{V}  \times \mathbb{Y} $.
\end{assum}
Under Assumption~\ref{ass:e-ioss}, the MHE optimizes the cost 
\begin{equation}
	\label{eq:cost_MHE}
	\begin{aligned}
			 J^\sharp(\hat{z}^\sharp_{k-N_k|k}, \hat{\boldsymbol{v}}_{\cdot|k}, \hat{\boldsymbol{y}}_{\cdot|k}) & = 
		     \sum^{N_k}_{j=1}\mu^{j} ( 2c_v \|\hat{v}_{k-j|k}-\bar{v}_{k-j}\|^2 + c_y \| y_{k-j}-\hat{y}_{k-j|k}\|^2 )\\
            & +\mu^{N_k} 2c_x \|\hat{z}^\sharp_{k-N_k|k}- \hat{z}^\sharp_{k-N_k}\|^2,  
	\end{aligned}
\end{equation}
with constants $\mu, c_x, c_y, c_v$  over variables $z^\sharp_{k-N_k|k}$,  $\hat{\boldsymbol{v}}_{\cdot|k}:=\bigl( \hat{v}_{j|k} \bigr)^{k-1}_{j=k-N_k}$ and $\hat{\boldsymbol{y}}_{\cdot|k}:=\bigl( \hat{y}_{j|k}\bigr)^{k-1}_{j=k-N_k}$.
The estimate $\hat{z}^\sharp_k$ for $z^\sharp_k$ is computed by 
\begin{equation}
\label{eq:prior}
	\hat{z}^\sharp_0= T_\sharp(\bar{x}_0) \ \text{and} \ \hat{z}^\sharp_k = \hat{z}^{\sharp,\star}_{k|k}, \ k\in \mathbb{I}_{\geq 1}, 	
\end{equation} 
where $\hat{z}^{\sharp,\star}_{k|k}$ is a minimizer to the following optimization problem
\begin{mini}{\hat{z}^\sharp_{k-N_k|k}, \hat{\boldsymbol{v}}_{\cdot|k}}{ J^\sharp(\hat{z}^\sharp_{k-N_k|k}, \hat{\boldsymbol{v}}_{\cdot|k}, \hat{\boldsymbol{y}}_{\cdot|k})  } {\label{eq:op_MHE}}{}
	\addConstraint{\hat{z}^\sharp_{j+1|k}}{=  \hat{f}_\sharp(\hat{z}^\sharp_{j|k}, y_j, u_j, \hat{v}_{j|k})}
	\addConstraint{ \hat{y}_{j|k}}{= \hat{h}_T(\hat{z}^\sharp_{j|k}, y_j,  \hat{v}_{j|k}), \hat{z}^\sharp_{k|k} \in \overline{\mathbb{X}}_\sharp }
	\addConstraint{\text{col}\big( }{\psi (y_j,  \hat{z}^\sharp_{j|k},  \hat{v}_{j|k}),  \hat{z}^\sharp_{j|k} \big) \in \overline{\mathbb{X}}}
	\addConstraint{\hat{v}_{j|k}}{ \in \mathbb{V},  \ \hat{y}_{j|k}\in \mathbb{Y}, \ j \in \mathbb{I}_{[k-N_k, k-1]}.}
\end{mini}
The functions $\hat{f}_\sharp$ and $\hat{h}_T$ in \eqref{eq:op_MHE} are specified in \eqref{eq:sys_red}. Further,   the function  $\psi$ satisfies the conditions in Assumption~\ref{ass}. 
Recall that $\overline{\mathbb{X}}$ is homeomorphic to $\mathbb{X}$ and $\overline{\mathbb{X} }_\flat \times \overline{\mathbb{X}}_\sharp \supset \overline{\mathbb{X}}$.  \blue{Thus,   the optimization problem~\eqref{eq:op_MHE} is always feasible.} 

In the following, we analyze the stability of  the proposed two-stage estimation scheme comprising \eqref{eq:fullest} and \eqref{eq:op_MHE}. 
\begin{thm}
\label{theo:stab}
	Suppose Assumptions~\ref{ass:compact}, \ref{ass} and~\ref{ass:e-ioss} apply. Let the prediction horizon  $N$ in \eqref{eq:op_MHE} satisfy 
	\begin{equation}
	\label{eq:N_cond}
	N >-\log_{\mu} (4c_x). 
	\end{equation} 
Then the estimation scheme consisting of   \eqref{eq:fullest},  \eqref{eq:op_MHE} \textcolor{black}{satisfies \eqref{eq:RGES} for all $k$, every trajectory of \eqref{eq:sys}, any priors $\bar{x}_0$, $\bar{v}_j$, $\bar{y}_j$ with $j\in \mathbb{I}_{[0,k]}$ and  $\bar{\boldsymbol{y}}_{[0, k-1]}=\boldsymbol{y}_{[0, k-1]}$}. In particular, it constitutes a UISE \textcolor{black}{according to Definition~\ref{def:RGAS_2}}.
\end{thm}
\begin{proof}
	In the proof, we first analyze bounds on $T_\sharp(x_k)-\hat{z}^\sharp_{k}$ with  $\hat{z}^\sharp_{k}$ from  MHE~\eqref{eq:op_MHE} by following reasoning similar to that used in  the proof of Theorem~\ref{theo:MHE_full} and then derive the estimation error of \eqref{eq:fullest}.  
	
	Recall that any trajectory $(\boldsymbol{x}, \boldsymbol{u}, \boldsymbol{w}, \boldsymbol{v}, \boldsymbol{y}) \in \mathbb{Z}^\infty$  of \eqref{eq:sys} induces a solution $(\boldsymbol{z}^{\sharp}, \boldsymbol{y}, \boldsymbol{u}, \boldsymbol{v}, \boldsymbol{y})  \in \mathbb{Z}^\infty_{\sharp}$  with $\boldsymbol{z}^{\sharp}= \bigl( T_\sharp(x_k)    \bigr)^{\infty}_{k=0}$  satisfying \eqref{eq:sys_red}. 
	Furthermore, due to the constraints in \eqref{eq:op_MHE},    $(\hat{z}_{j|k}^{\sharp,  \star}, y_j, u_j, \hat{v}^\star_{j|k}, \hat{y}^\star_{j|k}) \in \mathbb{Z}_{\sharp}$ is also a solution of \eqref{eq:sys_red} for $j\in \mathbb{I}_{[k-N_k, k-1]}$ at each time $k$. 
	Thus,  we can apply \eqref{eq:p-eioss_quadratic} successively together with \eqref{eq:cost_MHE}, \eqref{eq:prior} and $(a+b)^2 \leq 2a^2+2b^2$  to obtain 
	\begin{align*}
	\|T_\sharp(x_k)-\hat{z}^\sharp_k\|^2
	   &	\leq  \sum^{N_k}_{j=1} \mu^{j} 2 c_v \| v_{k-j} -\bar{v}_{k-j}\|^2 + 
		  \mu^{N_k} 2 c_x \| T_\sharp(x_{k-N_k}) - \hat{z}^{\sharp}_{k-N_k} \|^2 + J^\sharp(\hat{z}^{\sharp, \star}_{k-N_k | k}, \hat{\boldsymbol{v}}^\star_{\cdot | k}, \hat{\boldsymbol{y}}^\star_{\cdot|k}). 		
	\end{align*}
\noindent  From this condition,  we obtain  
\begin{equation}
	\label{eq:rges_proof_3}
	\begin{aligned}
		&\|T_\sharp(x_k)-\hat{z}^\sharp_k\|^2 \leq  \sum^{N_k}_{j=1} \mu^{j}  4 c_v \| v_{k-j} -\bar{v}_{k-j}\|^2  + \mu^{N_k} 4 c_x \| T_\sharp(x_{k-N_k}) - \hat{z}^{\sharp}_{k-N_k} \|^2 
	\end{aligned}
\end{equation}
by exploiting the  optimality of $J^\sharp$. 
Let us define $\rho:= \eta \hat{c}^{1/N}_x $,  $\tau:= k- \lfloor k/N \rfloor N$ with $\eta:=\sqrt{\mu}$,  $\hat{c}_x:=2\sqrt{c_x}$ and $\hat{c}_v:=2\sqrt{c_v}$. It follows from \eqref{eq:rges_proof_3} and $\|a\|^2+\|b\|^2\leq(\|a\|+\|b\|)^2$ that 
\begin{equation}
	\label{eq:rges_proof_4}
	\begin{aligned}
		 \|T_\sharp(x_k)-\hat{z}^\sharp_k\|  & \leq 
		\rho^{ \lfloor k/N \rfloor N}   \eta^\tau \hat{c}_x \| T_\sharp(x_0)-\hat{z}^\sharp_0 \| 
	 +  \rho^{ \lfloor k/N \rfloor N} \sum^{\tau}_{j=1} \eta^j \hat{c}_v \|v_{\tau-j}-\bar{v}_{\tau-j}\| \\
			& +\sum^{\lfloor k/N \rfloor-1}_{i=0} \rho^{iN}  \sum^{N}_{j=1}\eta^j \hat{c}_v \|v_{k-iN-j}-\bar{v}_{k-iN-j}\|  \\
	& \leq  \lambda^k \hat{c}_x \| T_\sharp(x_0)-\hat{z}^\sharp_0 \| + \sum^{k}_{i=1}  \lambda^i \hat{c}_v \| v_{k-i}-\bar{v}_{k-i}\|  ,
	\end{aligned}
\end{equation}
where $\lambda := \max(\rho, \eta) < 1$ in view of \eqref{eq:N_cond} and $\mu \in (0, 1)$. 
Let $L_{T_\sharp}$ and $L_{\psi}$ be the Lipschitz constants of $T_\sharp$ and  $\psi$ on compact sets $\mathbb{X}$ and $\mathbb{Y}  \times \overline{\mathbb{X}}_\sharp \times \mathbb{V}$ respectively.
\textcolor{black}{By leveraging \eqref{eq:fullest} and subadditivity of norms as well as   the fact that $\|\mathrm{proj}_{\overline{\mathbb{X}}} (z)-z\| \leq \|\bar{z}-z\| $ for any $\bar{z} \in \overline{\mathbb{X}}, z\in \mathbb{R}^n$,}  we obtain 
\begin{equation}
	\label{eq:rges_full_1}
\begin{aligned}
&	\|T(x_k)-T(\hat{x}_k)\| \leq \textcolor{black}{2}\|T(x_k)- \mathrm{col}\big( \psi(\textcolor{black}{\bar{y}}_k,  \hat{z}^\sharp_k,  \bar{v}_k), \hat{z}^\sharp_k \big)  \| \\
& \leq \textcolor{black}{ 2 \|T_\flat(x_k)-\psi(\textcolor{black}{\bar{y}}_k,  \hat{z}^\sharp_k,  \bar{v}_k)\|  + 2 \|T_\sharp(x_k)-\hat{z}^\sharp_k \|} \\
& \leq \blue{2\|\psi(y_k,  T_\sharp(x_k), v_k)-\psi(\blue{\bar{y}}_k,  \hat{z}^\sharp_k, \bar{v}_k)  \| + 2\|T_\sharp(x_k)-\hat{z}^\sharp_k \|} \\
&\leq  \blue{2L_{\psi}} (\|v_k-\bar{v}_k\| \blue{+\|y_k-\bar{y}_k\| })   + \blue{2\widetilde{L}_\psi} \|T_\sharp(x_k)-\hat{z}^\sharp_k \|  \\
 &\overset{\eqref{eq:rges_proof_4} }{\leq} \textcolor{black}{2\widetilde{L}_\psi}  \lambda^k \hat{c}_x L_{T_\sharp}\| x_0- \bar{x}_0\| + \textcolor{black}{2}\widetilde{L}_\psi \sum^{k}_{i=1}  \lambda^i  \hat{c}_v \| \delta v_{k-i}\| + \textcolor{black}{2}L_{\psi} (\|\delta v_k\|+\textcolor{black}{\|y_k-\bar{y}_k\|} ).    
\end{aligned}
\end{equation}
where $\widetilde{L}_{\psi}:= 1+ L_\psi$. 
Let $L_{T^{-1}}$ be the Lipschitz constant of $T^{-1}$ on some compact set. It follows from \eqref{eq:rges_full_1} that
\begin{align*}
	\|x_k-\hat{x}_k\| & \leq 2L_{T^{-1}} L_{\psi}(\|v_k-\bar{v}_k\| \blue{+\|y_k-\bar{y}_k\| }) 
+ 2\widetilde{L}_{\psi}  L_{T^{-1}} (  \lambda^k \hat{c}_x L_{T_\sharp}\| x_0- \bar{x}_0\| + \sum^{k}_{i=1}  \lambda^i  \hat{c}_v  \| v_{k-i}-\bar{v}_{k-i}\|)  \\
& \leq \lambda^k C_0 \| x_0- \bar{x}_0\| +   \sum^{k}_{i=0}  \lambda^i \big( C_v \| v_{k-i}-\bar{v}_{k-i}\| \blue{+ C_y \| y_{k-i}-\bar{y}_{k-i}\|}  \big) 
\end{align*}
with $C_v:=\blue{2L_{T^{-1}}} \max( L_{\psi},\widetilde{L}_{\psi}\hat{c}_v)$, \blue{$C_y:=2L_{T^{-1}} L_{\psi}$} and $C_0:= \blue{2\widetilde{L}_{\psi}}\hat{c}_x L_{T^{-1}} L_{T_\sharp}$. This leads to \eqref{eq:RGES}, thereby showing that the estimation scheme is UISE according to Definition~\ref{def:RGAS_2}. 
	\end{proof}

\begin{rem}
If $n^\sharp=0$, i.e., $n^\flat = n$, then \eqref{eq:fullest} reduces to $\hat{x}_k=T^{-1}\bigl(\textcolor{black}{\mathrm{proj}_{\overline{\mathbb{X}}}( \psi(\bar{y}_k, \bar{v}_k))}\bigr)$. 
By Lipschitz continuity, it immediately follows that the estimation error is bounded by $\|v_k -\bar{v}_k\|$ \textcolor{black}{and $\|y_k -\bar{y}_k\|$}.   Hence,  \eqref{eq:fullest} itself is a UISE, \textcolor{black}{and more specifically, a robust UISE.} 
If $n^\sharp=n$, then \eqref{eq:fullest} reduces to $\hat{x}_k= T^{-1}\big(\textcolor{black}{\mathrm{proj}_{\overline{\mathbb{X}}}(\hat{z}^\sharp_k)}\big)$.
\end{rem}

\blue{Throughout the paper,  the underlying optimization problems for MHE are assumed to be solved  to global optimality. Yet, this does not severely limit   applicability, since robustness to suboptimal solutions can be  enhanced by increasing horizon under   the proportional suboptimality concept  \cite{knufer2023nonlinear}. }

\section{Numerical example}\label{sec:sim}

To illustrate the benefits of our estimation scheme, we implement the proposed MHE-based UISE in Theorem~\ref{theo:MHE_full},  the two-stage UISE in Theorem~\ref{theo:stab}, and   the standard MHE scheme  from \cite{schiller2023} \blu{ on  a batch  chemical reaction process: $B \leftrightharpoons A+C$, $2A \leftrightharpoons C$, which is adapted from \cite[Example 4.39]{rawlings2020model}. The dynamics of the  concentrations $x^a$, $x^b$ and $x^c$ of species $A$, $B$ and $C$ is given by}
\begin{align*}
	\blu{x^a_{k+1}} &= \blu{x^a_{k} +
	\Delta_t \left( a_{1} x^b_{k} -a_{2} x^a_{k} x^c_{k} -2w_k + a_6 x^c_{k}\right)},\\
\blu{x^b_{k+1}} &= \blu{x^b_{k}-\Delta_t \left(a_{1} x^b_{k} - a_{2} x^a_{k} x^c_{k} \right),}   \\
\blu{x^c_{k+1}} &= \blu{x^c_{k}+\Delta_t \left(a_{1}x^b_{k} -a_{2} x^a_{k}x^c_{k} +w_k -a_{4}x^c_{k}  \right)},  \\
\blu{y^a_{k}} &= \blu{x^a_{k} +v^a_{k}, \ y^b_{k}=x^a_{k}+x^b_{k}+x^c_{k} +v^b_{k}}  
\end{align*}
\blu{with the sampling time $\Delta_t= 0.25 $.  The state $x_k=(x^a_{k}, x^b_{k}, x^c_{k})^\top$ is known to evolve in $\mathbb{X}=[0, 4]^3$.  The measurement noise $v_k= (v^a_{k},  v^b_{k})^\top$ is a uniformly distributed random variable satisfying $|v^i_{k}|\leq 0.01$ for $i=\{a,b\}$.  The unknown input $w_k$ models the output of the unknown function $W(x^a_{k},k)=(a_3+a_5 k) (x^a_k)^2 $.  The range of the unknown input  is unbounded and  is set to  $\mathbb{R}$.  The model parameters are chosen as follows:  
	$a_{1}=a_{3}=a_{6}=0.2$, $a_{2}=0.5$, $a_4 = 0.1$,  $a_{5} = 2$.}

\blu{The initial state is $x_{0}= (0.05, \textcolor{black}{0.5}, 0)^\top $ and its prior is chosen to be  $\bar{x}_{0}= (0.5, 1, 0.1)^\top$   for all three estimation schemes.    To design the UISE presented in Theorem~\ref{theo:MHE_full}, the condition in \eqref{eq:Lyapunov} is verified with $\mu=0.8$, a quadratic storage function $V(x,\tilde{x})= (x-\tilde{x})^\top P (x-\tilde{x})$ with
$$
P=\left(\begin{smallmatrix} 1.3656 & 0.0231 & -0.0188 \\ 0.0231 & 0.5817 & 0.6153 \\ -0.0188 & 0.6153 & 0.8877  \end{smallmatrix}\right), 
$$
	and the  scaled square functions $\alpha_1(r)=0.1r^2$, $\alpha_2=1.3688r^2$, $\sigma_v(r)=\sigma_y(r)=125r^2$. Consequently, the condition~\eqref{eq:MHE_full_horizon} holds with $\rho=0.9$ and any $N \geq 22$.      For  the two-stage UISE in Theorem~\ref{theo:stab},   we choose the map $T(x_k) = (x^a_{k}, x^b_{k}, 0.5x^a_{k} +x^c_{k})$, $T_\flat(x_k)=x^a_{k}$ and   $\psi(y_k,u_k,v_k)= y^a_{k}-v^a_{k}$, which  satisfy conditions in Assumption~\ref{ass}. With this choice, the condition \eqref{eq:p-eioss_quadratic} is verified with $\mu=0.8$, $c_x = 84.167$, $c_v=c_{\gamma}=125$ and $c_y=124.96$. Hence, the condition~\eqref{eq:N_cond} is satisfied with $N\geq 27$.   
The formulation  of  the standard MHE in \cite{schiller2023} is given by
\begin{mini*}{\hat{x}_{0 |k },  \hat{\boldsymbol{\omega}}_{\cdot |k } }{2 \eta^{N_k} \|\hat{x}_{k-N_k|k} - \hat{x}_{k-N_k}\|^2_{P_2} + \sum_{j=1}^{N_k} \eta^{j-1}(2\| \hat{\omega}_{k-j|k} \|^2_Q  +\|\hat{y}_{k-j|k}-y_{k-j}\|^2_R )   } {}{}
	\addConstraint{ \hat{x}_{j+1|k}}{=  f(\hat{x}_{j|k}, u_j, \hat{w}_{j|k}), \ j\in \mathbb{I}_{[k-N_k,k-1]}}
	\addConstraint{\hat{y}_{j|k}}{= h(\hat{x}_{j|k}, \hat{v}_{j|k}), j\in \mathbb{I}_{[k-N_k,k-1]}}
	\addConstraint{ \hat{y}_{j|k}}{ \in \mathbb{Y}, \  \hat{\omega}_{j|k}=(\hat{v}_{j|k}, \hat{w}_{j|k} ) \in \mathbb{V}\times \mathbb{R}^q,\ j\in \mathbb{I}_{[k-N_k,k-1]}  }
	\addConstraint{\hat{x}_{j|k}}{\in \mathbb{X}, \ j\in \mathbb{I}_{[k-N_k,k]}.  }
\end{mini*}
   For its implementation, the condition \cite[(4)]{schiller2023}  is verified  with $\eta=0.8$,  $Q=124.998 I_{3}$, $R=124.996 I_{2}$ and  
$$P_1=P_2= \left(\begin{smallmatrix} 0.7028 & -0.1351 & -0.1798 \\ -0.1351 & 2.145 & 2.6704 \\ -0.1798 & 2.6704 & 3.5871  \end{smallmatrix}\right).$$
Thus, the condition \cite[(16)]{schiller2023} implies that the horizon should be chosen such that   $N \geq 25$. 
For all estimation  schemes, we choose the prediction horizon} to  be $N=30\geq\operatorname{max}( \{22, 27, 25 \} )$.   \textcolor{black}{The proposed UISEs are implemented  with    measured outputs, i.e.,  $\bar{\boldsymbol{y}}_{[0,k]}=\boldsymbol{y}_{[0,k]}$.}
The simulation is performed on a PC with an Intel i7-12700K and 64GB of RAM. We use IPOPT \cite{waechter} and CasADi \cite{andersson} to implement MHE, and LMI solvers in Matlab to verify \blu{\eqref{eq:Lyapunov} and \eqref{eq:p-eioss_quadratic} as well as \cite[(4)]{schiller2023}.}  The code of numerical example is available at \url{https://github.com/memopanda/UISE.git}.  

\begin{figure}[ht]
	\centering
	\centering
\begin{subfigure}{0.41\textwidth}
	\vspace{20pt}
	\hspace{7pt}
	\includegraphics[width=0.92\textwidth, clip=false, trim = 16mm 3mm 15mm 12mm]{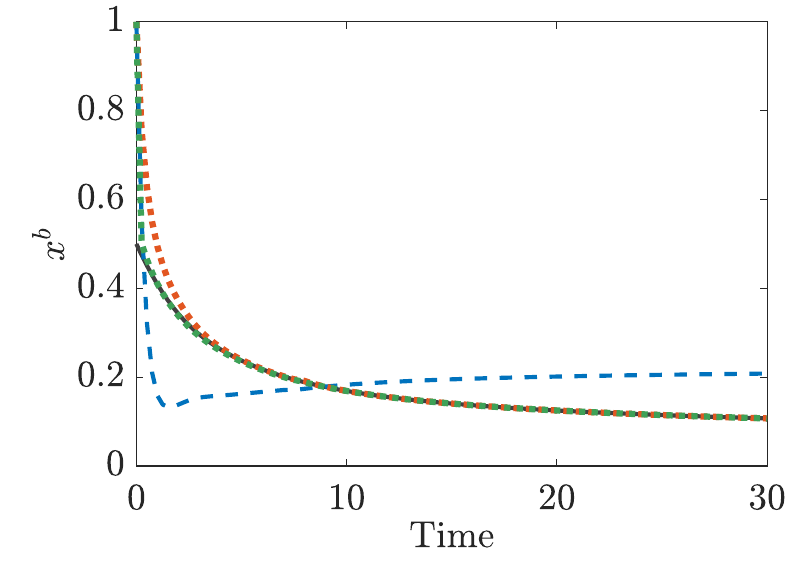}
	\caption{Noiseless case, i.e., $v = 0$  }
\end{subfigure}
\hspace{25pt}
\begin{subfigure}{0.41\textwidth}
	\vspace{20pt}
	\hspace{10pt}
	\includegraphics[width=0.92\textwidth, clip=false, trim = 16mm 3mm 15mm 12mm]{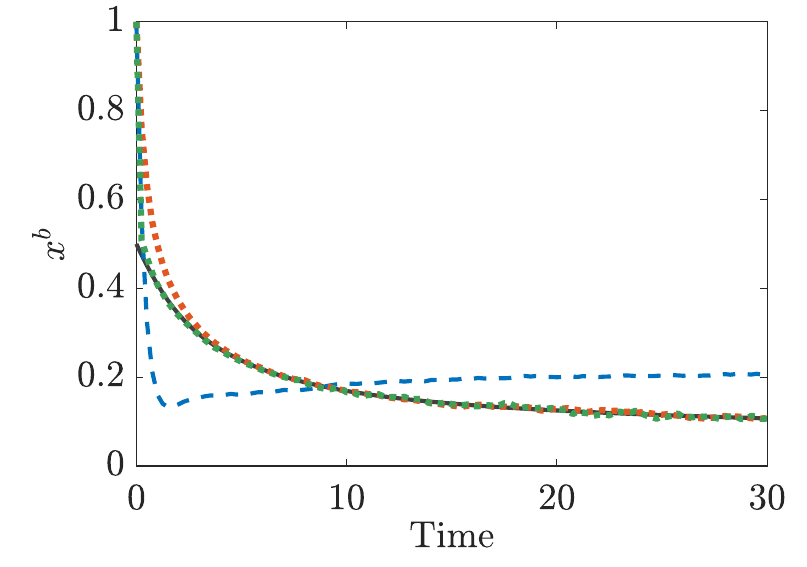}
	\caption{Noisy case, i.e., $v_k \neq0$ }
\end{subfigure}
\vspace{-1pt}
	\caption{\textcolor{black}{True value of $x^b$ (black) and  estimates obtained using the UISE in Theorem~\ref{theo:MHE_full} (green), the two-stage UISE in Theorem~\ref{theo:stab} (red) and the  standard MHE from \cite{schiller2023} (blue).}}
	\label{fig:traj}
\end{figure}

\begin{figure}
	\centering
	\begin{subfigure}{0.41\textwidth}
		\vspace{20pt}
		\hspace{7pt}
		\includegraphics[width=0.92\textwidth, clip=false, trim = 16mm 3mm 15mm 12mm]{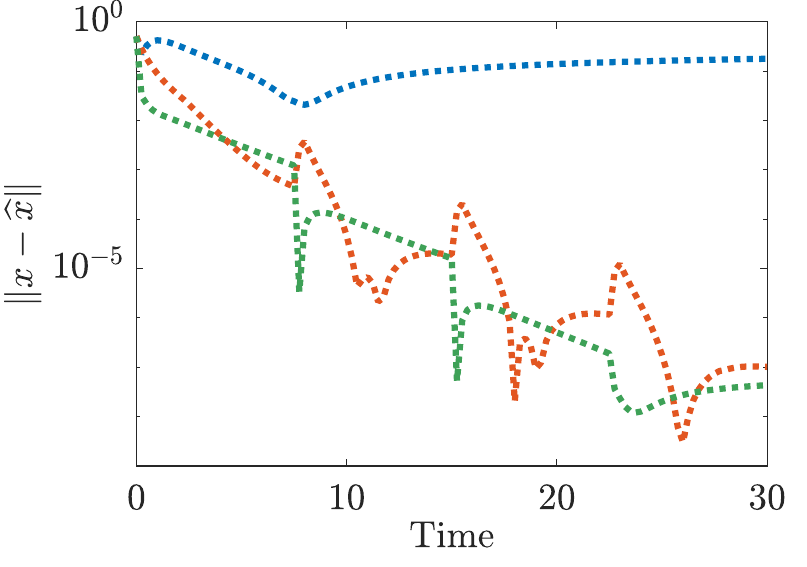}
		\caption{Noiseless case, i.e., $v = 0$  }
	\end{subfigure}
	\hspace{25pt}
	\begin{subfigure}{0.41\textwidth}
		\vspace{20pt}
		\hspace{10pt}
		\includegraphics[width=0.92\textwidth, clip=false, trim = 16mm 3mm 15mm 12mm]{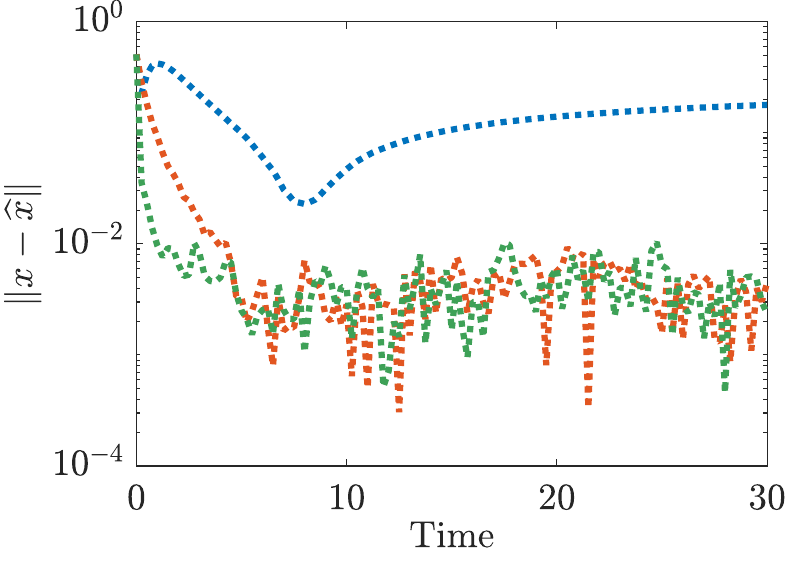}
		\caption{Noisy case, i.e., $v_k \neq0$ }
	\end{subfigure}
	\vspace{-1pt}
	\caption{Estimation errors  of the UISE in Theorem~\ref{theo:MHE_full} (green), the  two-stage UISE in Theorem~\ref{theo:stab} (red) and the standard MHE from \cite{schiller2023} (blue) }
	\label{fig:error_noises_all}
\end{figure}

\begin{figure}
	\vspace{20pt}
	\centering
	\includegraphics[width=0.38\textwidth, clip=false, trim = 16mm 3mm 15mm 12mm]{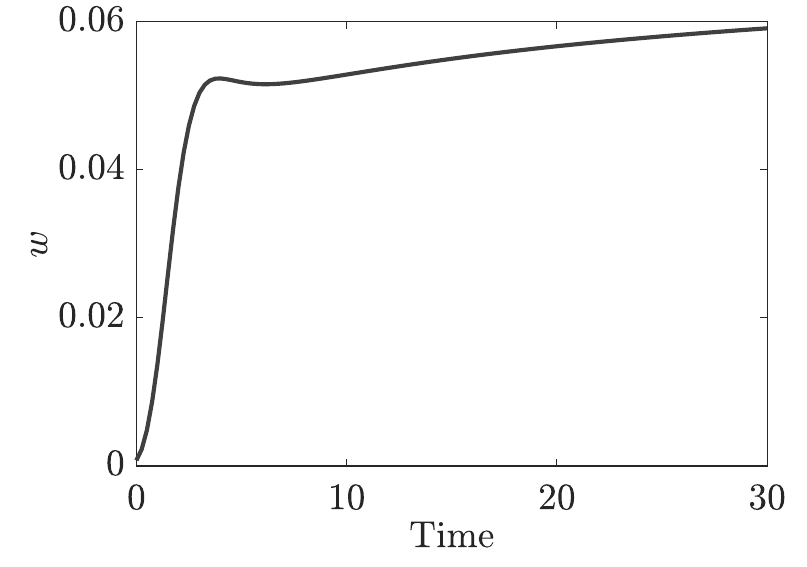}
	\caption{Unknown input $w$. }
	\label{fig:noises}
\end{figure}

 \blu{In both noiseless and noisy cases,  the estimation error of $x^b$ and the overall estimation error of the standard MHE are larger than those of the proposed UISEs,  as illustrated in Fig.~\ref{fig:traj}.}     The better estimation accuracy of the UISEs is attributed to  their capability to decouple the estimation errors from the unknown inputs.      Table~\ref{tab:table1} compares  the maximal average computation time $\mathbf{\Bar{t}}_{\mathrm{max}}$ for each sample, from 20 simulations,  among the different estimation schemes.   This clearly demonstrates  the  computational efficiency of the two-stage UISE compared with  the other two schemes. The computational efficiency is ascribed to   the smaller number of  optimization variables due to the use of   residual models in the design of  the two-stage UISE. 

\begin{table}[h!]\label{Tab:simulationtime}
	\begin{center}
		\caption{ $\mathbf{\Bar{t}}_{\mathrm{max}} (\text{ms})$ of  different estimation schemes.  }
		\label{tab:table1}
		\begin{tabular}{c c c c}
			\toprule 
			  \multirow{2}{4em}{noise $v_k$}  &  MHE  &  UISE  &   two-stage UISE \\
			  & \cite{schiller2023} & (Thm.~\ref{theo:MHE_full}) &  (Thm.~\ref{theo:stab}) \\
			\midrule 
			 $v_k\equiv 0$     & \blue{45.7}  & \blue{71.7} & \blue{23.4}    \\
			  $v_k\neq 0$     & \blue{47.1} &  \blue{75.7}  & \blue{24.9} \\    
			\bottomrule 
		\end{tabular}
	\end{center}
\end{table}

\section{Conclusions}\label{sec:conclusion}
To achieve  good estimation accuracy  for nonlinear systems subject to  unknown inputs, we have presented a design framework for \textcolor{black}{robust} UISEs that ensures convergent estimation error under vanishing \blue{measurement noise and output error}   despite possibly unbounded unknown inputs. We introduced the notion of strong nonlinear detectability, which admits an equivalent Lyapunov-like characterization. This allowed us to adopt the concept of moving horizon estimation to design an MHE-based \textcolor{black}{robust} UISE. Additionally, by properly deriving a residual model \textcolor{black}{and exploiting past output measurements},  we can reduce the computational complexity of MHE-based robust UISEs. A comparison with  a conventional MHE was performed on \blu{a batch chemical reaction process} to showcase the merits of the proposed approaches. As  future work, we will consider \blu{the simultaneous estimation of  states and unknown inputs as well as  explore the application of the resulting estimators to  predictive control of systems with unbounded unknown inputs.} 
Moreover, designing   non-asymptotic  unknown input state  estimators via moving horizon estimation will be addressed in future research.

\bibliographystyle{IEEEtran}        
\bibliography{references}           



\section*{Appendix}\label{sec:appendix}
\begin{prop} \label{prop:conv_max}
	Consider trajectories  $(\boldsymbol{x},  \boldsymbol{u}, \boldsymbol{w},  \boldsymbol{v}, \boldsymbol{y})$ and  
	$(\widetilde{\boldsymbol{x}},  \boldsymbol{u}, \widetilde{\boldsymbol{w}},  \widetilde{\boldsymbol{v}}, \widetilde{\boldsymbol{y}})$ of the system 
	\begin{subequations}
		\label{eq:sys_extend}
		\begin{align}
			x_{k+1} &= F(x_k, u_k, w_k, v_k) \\
			y_k &= H(x_k, u_k, w_k, v_k)
		\end{align}
	\end{subequations}
	with $x_k \in \mathbb{X}$, $u_k \in \mathbb{U}$, $w_k \in \mathbb{W}$, $v_k \in \mathbb{V}$ and $y_k \in \mathbb{Y}$.  There exist $\gamma_0$,  $\gamma_v, \gamma_y \in \mathcal{KL}$ such that  
	\begin{align*}
		\|\delta x_k\| \leq \max \Bigl( \gamma_0( \|\delta x_0 \|, k ), \gamma_v(  \| \delta \boldsymbol{v}\|_{[0,k-1]}),   
		 \gamma_y(\|\delta \boldsymbol{y}\|_{[0,k-1]}) \Bigr),   
	\end{align*}  
	if and only if there exist  $\overline{\gamma}_0$, $\overline{\gamma}_v$ , $\overline{\gamma}_y \in \mathcal{KL}$ such that 
	\begin{align*}
		\| \delta x_k \| \leq \max \Bigl(  \max_{j\in \mathbb{I}_{[0, k-1]}} \overline{\gamma}_v(\|\delta v_{j}\|, k-j-1), 
		 \overline{\gamma}_0( \|\delta x_0 \|, k ), \max_{j\in \mathbb{I}_{[0, k-1]}} \overline{\gamma}_y(\|\delta y_{j}\|, k-j-1) \Bigr),
	\end{align*}
	 where  $\diamond_i-\widetilde{\diamond}_i$ is denoted by $\delta\diamond_i$ for $\diamond \in \{ x,  v, y, w\}$. 
\end{prop}
\begin{proof}
	The proof follows a similar reasoning as in the proof of Proposition~2.5 in \cite{allan2021} and requires modifications of  the autonomous difference inclusion related to the original system in \cite{allan2021} for the more general system~\eqref{eq:sys_extend} to show the ``only if" statement. 
\end{proof}

\begin{prop}[Proposition~7 in \cite{SONTAG1998}] \label{prop:KL_lemma}
	Assume that $\beta \in \mathcal{KL}$. Then, there exist $\alpha_1$, $\alpha_2 \in \mathcal{K}_{\infty}$ so that 
	$$ \beta(r,k) \leq \alpha_1\big( \alpha_2(r) e^{-k} \big), \ \forall r, k \geq 0.   $$ 
\end{prop}

\input{UISE_with_fullmodel_proof}

\begin{prop} \label{prop:weak_triangular}
	For any function $\beta$ of class $\mathcal{GK}$, any function $\alpha$ of class $\mathcal{K}_{\infty}$ such that $\alpha-\text{Id}$ is of class $\mathcal{K}_{\infty}$, and any  $a, b \in \mathbb{R}_{\geq 0}$, we have
	\begin{equation}
		\label{eq:weak_triangular}
		\beta(a+b) \leq \beta(\alpha(a)) + \beta(\alpha\circ (\alpha-\text{Id})^{-1}(b)).
	\end{equation}
\end{prop}
\begin{proof}
	Choose any function $\alpha \in \mathcal{K}_{\infty}$ such that $\alpha-\text{Id} \in \mathcal{K}_{\infty}$. For any $a, b \in \mathbb{R}_{\geq 0}$, if
	$ (\alpha-\text{Id})(a) \leq b$, then $(\alpha-\text{Id})^{-1}(b) \geq a$. Hence,  we have
	$$ b+(\alpha-\text{Id})^{-1}(b) = \alpha\circ(\alpha-\text{Id})^{-1} (b) \geq a+b $$
	As any function in class $\mathcal{GK}$ is non-decreasing, the above inequality implies that  $\beta(a+b) \leq \beta(\alpha\circ (\alpha-\text{Id})^{-1}(b))$, leading to \eqref{eq:weak_triangular}. 
	Similarly, if  $ (\alpha-\text{Id})(a) > b$, then
	$$ a+b <  a+(\alpha-\text{Id})(a)=\alpha(a).  $$ 
	As $\beta$ is  non-decreasing, we have $\beta(a+b) \leq \beta(\alpha(a))$, and hence showing \eqref{eq:weak_triangular}. 
\end{proof}
\end{document}

%% file: UISE_with_fullmodel_proof.tex
\begin{pf*}[Proof of Theorem~\ref{theo:Lyapunov}:] 
	To ease  notation, we abbreviate $\diamond_i-\widetilde{\diamond}_i$ by $\delta\diamond_i$ for $\diamond \in \{ x,  v, y, v^e, y^e\}$ in the proof. 
	
\textbf{Sufficiency}: Let us consider  trajectories  $(\boldsymbol{x},  \boldsymbol{u}, \boldsymbol{w},  \boldsymbol{v}, \boldsymbol{y})$, 
$(\widetilde{\boldsymbol{x}},  \boldsymbol{u}, \widetilde{\boldsymbol{w}},  \widetilde{\boldsymbol{v}}, \widetilde{\boldsymbol{y}})$ $ \in $ $  \mathbb{Z}^\infty$ satisfying  \eqref{eq:sys}.  
Successive application of  \eqref{eq:Lyapunov_2} leads to 
\begin{align*}
	V(x_k, \widetilde{x}_k) -   \mu^k V(x_0, \widetilde{x}_0) &\leq    
   \sum^{k-1}_{i=0} \mu^{k-i-1} \big( \sigma_v(\|\delta v_{i+1}\|) + \sigma_y(\|\delta  y_{i+1}\|) +  \sigma_y(\|\delta y_{i}\|)\big)\\
   & +\sum^{k-1}_{i=0} \mu^{k-i-1} \sigma_v(\|\delta v_i\|).
\end{align*}
Since $\mu\in (0,1)$, we further obtain
\begin{equation}
\label{eq:lyapunov_proof_1}
\begin{aligned}
	 V(x_k, \widetilde{x}_k) -   \mu^k V(x_0, \widetilde{x}_0) & \leq  \sum^{k-1}_{i=0} \mu^{k-i-1} \sigma_v(\|\delta v_i\|)   
   + \sum^{k-1}_{i=0} \mu^{k-i-1}  \sigma_y(\|\delta y_i\|)\\
   &+ \sum^{k-1}_{i=0} \mu^{k-i-2}\big( \sigma_v(\| \delta v_{i+1}\|) +  \sigma_y(\| \delta y_{i+1}\|)  \big)\\
&\leq  \mu^{-1}\alpha_v(\|\delta v_k\|) +  \mu^{-1}\alpha_y(\|\delta y_k\|) 
 +\sum^{k-1}_{i=0}2\mu^{k-i-1}\big( \sigma_v(\|\delta v_i\|) + \sigma_y(\|\delta y_i\|) \big)\\
&\leq  \sum^{k}_{i=0}2\mu^{k-i-1}\big( \sigma_v(\|\delta v_i\|) + \sigma_y(\|\delta y_i\|) \big).
\end{aligned}
\end{equation}
By combining the lower and upper bounds of $V(x, \widetilde{x})$ in \eqref{eq:Lyapunov_1} with \eqref{eq:lyapunov_proof_1} and then choosing $\alpha = \alpha_1$, $\alpha_0(r,i)=\mu^i\alpha_2(r)$, $\alpha_v(r,i)=2\mu^{i-1}\sigma_v(r)$ and  $\alpha_y(r,i)=2\mu^{i-1}\sigma_y(r)$, we arrive at \eqref{eq:strong_detect}. Additionally, both $\alpha_v$ and $\alpha_y$ are summable $\mathcal{KL}$-functions by  the formula of infinite geometric series.    

\textbf{Necessity}: Consider the system \eqref{eq:sys_1} with the auxiliary output $y_{k+1}=h(f(x_k, u_k, w_k, v_k), v_{k+1})$,  and define $y^e_k:=\text{col}(y_k, y_{k+1})$,  $v^e_k:=\text{col}(v_k, v_{k+1})$,  $\widetilde{y}^e_k:=\text{col}(\widetilde{y}_k, \widetilde{y}_{k+1})$ and $\widetilde{v}^e_k:=\text{col}(\widetilde{v}_k, \widetilde{v}_{k+1})$, then \eqref{eq:strong_detect} implies that the the system with the auxiliary output satisfies
\begin{equation}
\label{eq:lypaunov_proof_3_1}
\begin{aligned}
	\alpha(\|\delta x_k\|) & \leq  \alpha_0( \|\delta x_0 \|, k ) +   \alpha_v (\|\delta v_{0}\| , k ) + \alpha_y (\|\delta y_{0}\| , k ) \\
& + \sum_{i=0}^{k-1} \bigl(  \alpha_v (\|\delta v^e_{i}\| , k-i-1 ) + \alpha_y (\|\delta y^e_{i}\| , k-i-1 ) \bigr) \\
& \leq \alpha_0( \|\delta x_0 \|, k )+\alpha_v (\|\delta v^e_{0}\| , k-1 ) +  \alpha_y (\|\delta y^e_{0}\| , k-1 ) \\
&  + \sum_{i=0}^{k-1} \bigl(  \alpha_v (\|\delta v^e_{i}\| , k-i-1 ) + \alpha_y (\|\delta y^e_{i}\| , k-i-1 ) \bigr)  \\
& \leq \sum_{i=0}^{k-1} \big(  2 \alpha_v (\| \delta v^e_{i}\| , k-i-1 ) + 2 \alpha_y (\|\delta y^e_{i}\| , k-i-1 ) \big)
+\alpha_0( \|\delta x_0 \|, k ) . 
\end{aligned}
\end{equation}
Recall that $\alpha_v$ and $\alpha_y$ are summable $\mathcal{KL}$-functions, this together with \eqref{eq:lypaunov_proof_3_1} implies that  there exist $\overline{\alpha}_v, \overline{\alpha}_y \in \mathcal{K}$, such that 
\begin{align*}
	\alpha(\|\delta x_k\|) \leq &  \alpha_0( \|\delta x_0 \|, k ) + \overline{\alpha}_v( \|\delta \boldsymbol{v}^e\|_{[0,k-1]}) 
	+ \overline{\alpha}_y(\|\delta \boldsymbol{y}^e\|_{[0,k-1]}).   
\end{align*}
As a result,  we have
\begin{equation}
	\label{eq:lypaunov_proof_3_2}
	\begin{aligned}
		\|\delta x_k\| \leq &\max \Bigl( \gamma_0( \|\delta x_0 \|, k ), \gamma_v(  \|\delta \boldsymbol{v}^e\|_{[0,k-1]}), 
		 \gamma_y(\|\delta \boldsymbol{y}^e\|_{[0,k-1]}) \Bigr),   
	\end{aligned}
\end{equation}
with $\gamma_0:= \alpha^{-1}\circ 3 \alpha_0 \in \mathcal{KL}$, $\gamma_v:= \alpha^{-1} \circ 3\overline{\alpha}_v \in \mathcal{K}$, $ \gamma_y:=\alpha^{-1} \circ 3\overline{\alpha}_y \in \mathcal{K}$. By invoking Proposition~\ref{prop:conv_max},    
\eqref{eq:lypaunov_proof_3_2} can be reformulated into
\begin{equation}
	\label{eq:lyapunov_proof_3_3}
	\begin{aligned}
		\| &\delta x_k \| \leq \max \Bigl(  \max_{j\in \mathbb{I}_{[0, k-1]}} \overline{\gamma}_v(\|\delta v^e_{j}\|, k-j-1),  
		 \overline{\gamma}_0( \|\delta x_0 \|, k ), \max_{j\in \mathbb{I}_{[0, k-1]}} \overline{\gamma}_y(\|\delta y^e_{j}\|, k-j-1) \Bigr),
	\end{aligned}
\end{equation} 
with some $\overline{\gamma}_0,  \overline{\gamma}_v, \overline{\gamma}_y \in  \mathcal{KL}$. 
By applying Proposition~\ref{prop:KL_lemma} to \eqref{eq:lyapunov_proof_3_3} along with $\max(a,b)\leq a+b$ for $a, b \geq 0$, and then denoting  $\lambda:=e^{-1}$,  we can find $\beta, \beta_0, \beta_v, \beta_y \in \mathcal{K}^\infty$ such that
\begin{equation}
	\label{eq:lyapunov_proof_3}
\begin{aligned}
\beta(\|\delta x_k\| ) & \leq   \sum^{k-1}_{i=0}\lambda^{k-i-1} \bigl( \beta_v(\|\delta v_{i+1}\|) +\beta_y(\|\delta y_{i+1}\|) \bigr)  \\
& +  \sum^{k-1}_{i=0} \lambda^{k-i-1} \big(\beta_v(\|\delta v_i\|)+  \beta_y(\|\delta y_{i}\|) \big) +  \lambda^k\beta_0(\|\delta x_0\|). 
\end{aligned}
\end{equation} 
Let us    abbreviate two solutions  $(\boldsymbol{x},  \boldsymbol{u}, \boldsymbol{w},  \boldsymbol{v}, \boldsymbol{y})$ and $	(\widetilde{\boldsymbol{x}},  \boldsymbol{u}, \widetilde{\boldsymbol{w}},  \widetilde{\boldsymbol{v}}, \widetilde{\boldsymbol{y}})$ of \eqref{eq:sys} by $ \boldsymbol{\Xi}$ and $\widetilde{\boldsymbol{\Xi}}$ respectively.     For a given  $(x, \widetilde{x}) \in \mathbb{X} \times \mathbb{X}$, let us define 
\begin{align*}
 V(x, \widetilde{x}) &:= \sup_{ \substack{   \bar{k}\in \mathbb{I}_{\geq 0}, \ \boldsymbol{\Xi}, \widetilde{\boldsymbol{\Xi}}  \in \mathbb{Z}^{\infty}, \\  
	\ x_0 = x, \ \widetilde{x}_0=\widetilde{x} } }  \lambda^{-\bar{k}/2} \Bigl( \beta(\|\delta x_{\bar{k}}\|) 
  - \sum^{\bar{k}}_{i=1}\lambda^{\bar{k}-i} \bigl(  \beta_v(\| \delta v_i\|) + \beta_y(\|\delta y_{i}\|) \bigr) \\
&- \sum^{\bar{k}-1}_{i=0}\lambda^{\bar{k}-i-1} \bigl( \beta_v(\|\delta v_i\|) + \beta_y(\|\delta y_{i}\|) \bigr)  \Bigr).
\end{align*}
By leveraging this definition together with \eqref{eq:lyapunov_proof_3} and   following a similar reasoning as in the proof of \cite[Theorem~3.2]{allan2021}, we obtain 
\begin{align*}
	\beta(\|x-\widetilde{x}\|) \leq V(x, \widetilde{x}) \leq \beta_0(\|x-\widetilde{x}\|)
\end{align*}
and
\begin{align*}
	V(x^+, \widetilde{x}^+) & \leq \sqrt{\lambda} V(x, \widetilde{x}) + \beta_v(\|v-\widetilde{v}\|) +\beta_v(\| v^+ - \widetilde{v}^+\|)
 + \beta_y(\|y-\widetilde{y}\|) + \beta_y(\|y^+-\widetilde{y}^+\|), 
\end{align*}
with $y, y^+, v, v^+$  specified in Theorem~\ref{theo:Lyapunov} and  $x^+:=f(x,u,w)$ as well as $\widetilde{x}^+:=f(\widetilde{x}, u, \widetilde{w})$. 
\end{pf*}